\documentclass[12pt,british]{elsarticle}
\usepackage[T1]{fontenc}
\usepackage[utf8]{inputenc}
\usepackage{xcolor}
\usepackage[british]{babel}
\usepackage{array}
\usepackage{float}
\usepackage{mathtools}
\usepackage{multirow}
\usepackage{dsfont}
\usepackage{amsmath}
\usepackage{amsthm}
\usepackage{amssymb}
\usepackage{graphicx}
\usepackage[a4paper]{geometry}
\usepackage[bookmarks=true,bookmarksnumbered=false,bookmarksopen=false,
 breaklinks=false,pdfborder={0 0 1},backref=false,colorlinks=false]
 {hyperref}
\hypersetup{pdftitle={DOE final report}}

\makeatletter

\providecommand{\tabularnewline}{\\}

\theoremstyle{plain}
\newtheorem{thm}{\protect\theoremname}

\usepackage{url}

\usepackage{rotating}

\usepackage{amsthm}

\usepackage{bbm}

\usepackage{paracol}

\usepackage{subfig}

\usepackage[style=british]{csquotes}%

\usepackage{caption}

\ifdefined\showcaptionsetup
 \PassOptionsToPackage{caption=false}{subfig}
\fi
\usepackage{subfig}
\makeatother

\addto\captionsamerican{\renewcommand{\theoremname}{Theorem}}
\addto\captionsbritish{\renewcommand{\theoremname}{Theorem}}
\addto\captionsenglish{\renewcommand{\theoremname}{Theorem}}
\providecommand{\theoremname}{Theorem}

\usepackage[numbers]{natbib}
\addto\captionsbritish{}
\begin{document}
\begin{frontmatter}
\title{\textbf{Characterisation of conserved and reacting moieties in chemical
reaction networks}}
\author[1,2]{Hadjar Rahou}
\author{Hulda S. Haraldsdóttir}
\author[1,2]{Filippo Martinelli}
\author[1,2,3,4]{Ines Thiele} 
\author[1,2]{Ronan~M.T.~Fleming\corref{corresponding}}

\cortext[corresponding]{Corresponding author. Email: \href{mailto:ronan.mt.fleming@gmail.com}{ronan.mt.fleming@gmail.com}}

\address[1]{School of Medicine, University of Galway, Ireland.}
\address[2]{Digital Metabolic Twin Centre, University of Galway, Ireland.}
\address[3]{School of Microbiology, University of Galway, Galway, Ireland.}
\address[4]{APC Microbiome Ireland, Cork, Ireland.}



\begin{abstract}
A detailed understanding of biochemical networks at the molecular
level is essential for studying complex cellular processes. In this
paper, we provide a comprehensive description of biochemical networks
by considering individual atoms and chemical bonds. To address combinatorial
complexity, we introduce a well-established approach to group similar
types of information within biochemical networks. A conserved moiety
is a set of atoms whose association is invariant across all reactions
in a network. A reacting moiety is a set of bonds that are either
broken, formed, or undergo a change in bond order in at least one
reaction in the network. By mathematically identifying these moieties,
we establish the biological significance of conserved and reacting
moieties according to the mathematical properties of the stoichiometric
matrix. We also present a novel decomposition of the stoichiometric
matrix based on conserved moieties. This approach bridges the gap
between graph theory, linear algebra, and biological interpretation,
thus opening up new horizons in the study of chemical reaction networks.

\vspace{0.2cm}


\vspace{0.2cm}

\end{abstract}
\end{frontmatter}

\section{Introduction}

\paragraph{Mathematical representation of biochemical networks}

Mathematical analysis of biochemical networks enables one to identify
novel characteristics of biochemical networks and define biological
concepts in terms of mathematical objects. One approach that enables
this study is to represent the stoichiometry of a biochemical network
by a stoichiometric matrix. A \textit{stoichiometric matrix} is a
rectangular matrix where each row represents a molecular species and
each column represents a reaction. Typically there are more reactions
than molecular species. Each entry in a stoichiometric matrix is given
by the integer stoichiometric coefficient of a molecular species in
a reaction, which is negative if a molecular species is a substrate
and positive if a molecular species is a product in that reaction.
A biochemical network can be represented as a hypergraph. In a hypergraph,
each vertex represents a biochemical species, while each hyperedge
represents a reaction that connects multiple species. Unlike a simple
graph, where edges connect only two nodes, hyperedges may connect
more than two nodes, reflecting the complex interactions in biochemical
reactions involving multiple reactants and products. 

\paragraph{Stoichiometric matrix \& molecular topology }

Molecular topology only considers the connectivity of atoms (i.e.,
which atoms are bonded to which) and it does not inherently capture
spatial arrangements like stereochemistry. From a stoichiometric matrix
alone, one cannot derive the molecular topology of each species in
the underlying biochemical network. This statement is obvious, but
it implies that one cannot obtain a biochemically faithful mathematical
representation of a biochemical network without incorporating a representation
of molecular topology into ones mathematical analysis of a biochemical
network. Previously, we demonstrated that incorporation of molecular
species topology in the form of a graph, where each vertex is an atom
of a specific element and each edge is a bond between atoms, enables
identification of a set of conserved moiety vectors \cite{haraldsdottir_Identification_2016},
each of which is interpretable in terms of a structurally defined
conserved moiety. Subsequently, we demonstrated that by considering
species topology a stoichiometric matrix may be split into the sum
of $m-\textrm{rank}(N)$ moiety transition matrices, each of which
corresponds to a subnetwork corresponding to a structurally identifiable
conserved moiety. 

Identification of conserved moieties provides detailed information
about invariant sets of atoms in a metabolic network, but it does
not directly consider bonds between atoms. In a chemical reaction,
typically, only a few atoms directly participate in broken or formed
bonds. In the literature, different terms are used for the part of
a molecular species that changes in a chemical reaction. The reaction
centre is defined as the atoms and bonds that are directly involved
in the bond and electron rearrangement of a reaction \cite{chen_automatic_2013}.
Elsewhere the reaction site is defined as a subtopology that includes
the reaction centre \cite{funatsu_automatic_1988}. There are several
different approaches to finding reaction centres. These include computational
methods such as molecular dynamics simulations \cite{hollingsworth_molecular_2018}.
From our perspective, a weakness of existing approaches is that reaction
centres are defined heuristically or computationally in a manner that
does not admit an unambiguous mathematical interpretation.

Current computational methods to identify reaction centres are primarily
based on identifying the maximum common subtopology between a molecular
species and its product pair. While these methods are useful, they
are typically designed to handle specific biochemical reactions and
are not well-suited for genome-scale models. Similarly, automatic
identification of reaction rules involves analysing large databases
of chemical reactions to identify patterns in how different functional
groups react. However, this approach may not fully capture the complexity
of underlying chemical networks, particularly in large-scale systems.
This highlights the need for a new method capable of handling genome-scale
models with greater accuracy and scalability.

\paragraph{Aims and Outline}

Herein, we deepen our investigation of the intersection between stoichiometric
matrices, molecular topology, and graph theory, considering both the
atoms and bonds involved in each reaction. In descriptive terms, a
\textit{conserved moiety} is a set of atoms that remains intact in
a reaction network, while a \textit{reacting moiety} is a set of reacting
chemical bonds between a pair of conserved moiety instances that dissociate
in at least one reaction of a reaction network. We mathematically
define conserved and reacting moieties in terms of invariant and variant
subsets of an atom transition graph, where each vertex corresponds
to an atom whose transition from substrate to product either corresponds
to an unbroken or broken bond in a reaction. Furthermore, we present
a novel and efficient algorithm to identify conserved and reacting
moieties given a stoichiometric matrix, a \textit{molecular graph}
for each molecular species, and an atom mapping for each reaction.

Moreover, we tackle the challenge of complexity reduction in biochemical
networks by proposing a novel decomposition of the stoichiometric
matrix in terms of conserved moieties. This moiety decomposition is
a simplification that reflects the participation of every molecular
species in every reaction in which it participates in within a given
biochemical network. We also discuss future directions and potential
challenges, including expanding this moiety decomposition to larger
networks. These contributions offer valuable insights into the functional
aspects and network topology of biochemical systems, advancing our
understanding of complex biological processes. The mathematical results
are illustrated using a toy example reaction network introduced previously
in \cite{ghaderi_structural_2020}.\medskip{}

All symbols used in this paper are outlined in Table \ref{tab:Graph-theory-notation},
which includes notations relevant to graph theory, Table \ref{tab:Matrix-notations-for},
detailing matrices used throughout the paper, and Table \ref{tab:Variable-notations-for},
listing variables and counts (See Section \ref{sec:Notation-tables}
in the supplemental material).

\section{Mathematical Foundations}

We first introduce the essential mathematical concepts that form the
foundation of the main results of this work.

\subsection{Graph and hypergraph}

A graph is a set of vertices and a set of edges, where each edge connects
exactly two distinct vertices. In contrast, a hypergraph generalises
this concept by allowing hyperedges to connect any number of vertices,
making it a suitable representation for complex biochemical reactions
involving multiple reactants and products simultaneously\cite{klamt_hypergraphs_2009}. 

\subsection{\label{subsec:Graph-isomorphism-and}Graph isomorphism and isomorphism
classes}

An isomorphism of graphs $\mathcal{G}_{A}$ and $\mathcal{G}_{B}$
is a bijection between the vertex sets of $\mathcal{G}_{A}$ and $\mathcal{G}_{B}$
denoted $f:\mathcal{X}(\mathcal{G}_{A})\rightarrow\mathcal{X}(\mathcal{G}_{B})$
such that any two vertices $\mathcal{X}_{i}$ and $\mathcal{X}_{j}$
of $\mathcal{G}_{A}$ are adjacent in $\mathcal{G}_{A}$ if and only
if $f(\mathcal{X}_{i})$ and $f(\mathcal{X}_{j})$ are adjacent in
$\mathcal{G}_{B}$. That is, there exists a permutation matrix $P$
such that $A=PBP^{T}$, where $A$ and $B$ denote the incidence matrices
representing the graphs $\mathcal{G}_{A}$ and $\mathcal{G}_{B}$
respectively. A label-preserving graph isomorphism occurs when two
graphs are permutationally equivalent, as above, and the labels on
the vertices are preserved. A \textit{graph isomorphism class }is
a set of graphs that are all isomorphic to each other. A maximal subgraph
isomorphism class of a graph is a maximal set of pairwise isomorphic
connected components.

\subsection{Graph splitting}
\begin{thm}
\label{thm:splittingN} Let $A\in\left\{ -1,0,1\right\} ^{p\times q}$
be an incidence matrix for a graph $\mathcal{A}(\mathcal{X},\mathcal{\mathcal{E}},\mathcal{H})$.
Let $C\in\{0,1\}^{c\times p}$ be a mapping between connected components
and vertices in a graph, where $C_{i,j}=1$ if connected component
$i$ contains vertex $j$ and $C_{i,j}=0$ otherwise, then $c=p-\textrm{rank}(A)$,
$CA=0$ and the following matrix splitting exists
\begin{equation}
A=\mathrm{diag^{-1}}\left(C^{T}\mathbbm{1}\right)\sum_{i=1}^{c}A(i),\label{eq:Asplit}
\end{equation}
where $A(i)\in\left\{ -1,0,1\right\} ^{p\times q}$ is an incidence
matrix for the $i^{th}$ connected component of $\mathcal{G}(\mathcal{X},\mathcal{\mathcal{E}},\mathcal{H})$,
given by
\begin{equation}
A(i)\coloneqq\textrm{diag}(C_{i,:})A\label{def:A(k)}
\end{equation}
\end{thm}

\begin{proof}
That $c=p-\textrm{rank}(A)$ and $CA=0$ are standard results from
algebraic graph theory (Theorem 2.5 \cite{grossman_minors_1995}).
Substituting (\ref{eq:Asplit}) into (\ref{def:A(k)}), it is enough
to show $C^{T}\mathbbm{1}\in\mathbb{Z}_{++}^{m}$ and that 
\[
\textrm{diag}\left(C^{T}\mathbbm{1}\right)=\sum_{i=1}^{c}\textrm{diag}(C_{i,:}).
\]
The expression on the left sums each row of $C$ then places it on
the diagonal of an $p\times p$ matrix. The expression on the right
places each row of $C$ on the diagonal of a matrix, and sums the
matrices, which is equivalent to the expression on the left as the
operations involved are commutative. Each entry of $C$ is non-negative
so $C^{T}\mathbbm{1}\ge0,$ therefore it remains to show that $C^{T}\mathbbm{1}\in\mathbb{Z}_{++}^{m}$.
Every atom is part of one connected component, so $C^{T}\mathbbm{1}>0$,
giving the desired result.
\end{proof}

\subsection{\label{subsec:Graph-condensation}Graph condensation }

Graph condensation is a process that reduces a graph by merging a
set of vertices into a single vertex based on certain criteria, typically
to simplify the analysis of complex networks \cite{balcer_computing_1973}.
Given an undirected graph $\mathcal{G}$, its condensation, denoted
$\mathcal{G}^{c}$, is obtained by contracting each connected component
of $\mathcal{G}$ into a single vertex. Each vertex of $\mathcal{G}^{c}$
corresponds to a connected component of the original graph.

\subsection{Set cover problem}

Given a set of elements $\mathcal{E}=\{e_{1},e_{2},\ldots,e_{n}\}$
and a set of $m$ subsets of that set, ${\color{black}\mathcal{S}(\mathcal{E})=\{\mathcal{S}_{1},\mathcal{S}_{2},\ldots,\mathcal{S}_{{\color{violet}{\color{black}m}}}}\}$,
the set cover problem is to find a minimal collection $\mathcal{C}$
of sets from $\mathcal{S}$ such that $\mathcal{C}$ covers all elements
in $\mathcal{E}$. That is $\cup_{\mathcal{S}_{i}\in\mathcal{C}}\mathcal{S}_{i}=\mathcal{E}$.
The set cover problem is a classic NP-hard problem where the objective
is to cover a universal set $\mathcal{E}$ with the smallest number
of subsets from a collection $\mathcal{S}$ \cite{vazirani_approximation_2003}\cite{feige_threshold_1998}.
Due to its complexity, several algorithms are used to find feasible
solutions. The greedy algorithm \cite{chvatal_greedy_1979} is widely
applied due to its simplicity and effectiveness. It iteratively selects
the subset that covers the largest number of uncovered elements, achieving
a near-optimal approximation ratio of $\ln|\mathcal{E}|$, which is
among the best possible for polynomial-time solutions. A linear programming
(LP) relaxation offers another efficient approach by solving a fractional
version of the problem. The fractional solution is then converted
to integer form using rounding techniques like randomised rounding,
or threshold rounding, allowing for flexibility in handling weighted
instances while maintaining solution quality. Primal-dual algorithms
construct solutions by simultaneously adjusting primal and dual variables,
yielding good approximations with efficiency suited to large-scale
or dynamically evolving problems. For particularly large or complex
instances, metaheuristics such as genetic algorithms and simulated
annealing provide flexible, high-quality solutions without guaranteeing
optimality, making them useful for problem-specific constraints and
large datasets.

\section{Hypergraph and graph representations of a metabolic network}

The following section introduces hypergraph and graph representations
of a metabolic network, which are the foundation for mathematical
analysis of metabolism.

\subsection{\label{subsec:Directed-stoichiometric-hypergra}Directed stoichiometric
hypergraph}

A metabolic network is represented by a directed stoichiometric hypergraph
$\mathcal{H}(\mathcal{V},\mathcal{Y}(\mathcal{S},\mathcal{P}))$,
which is an oriented hypergraph that consists of a sequence of $m$
vertices $\mathcal{V}:=(\mathcal{V}_{1},\ldots,\mathcal{V}_{m})$,
and a sequence of $n$ directed hyperedges $\mathcal{Y}:=(\mathcal{Y}_{1},\ldots,\mathcal{\mathcal{Y}}_{n})$.
In the $j^{th}$ reaction $\mathcal{Y}_{j}\coloneqq(\mathcal{S}_{j},\mathcal{P}_{j})$
the substrate (arrow tail) complex is
\[
\mathcal{S}_{j}\coloneqq\sum_{i=1}^{m}F_{i,j}\mathcal{V}_{i}
\]
and the product (arrow head) complex is
\[
\mathcal{P}_{j}\coloneqq\sum_{i=1}^{m}R_{i,j}\mathcal{V}_{i}
\]
where $F\in\mathbb{Z}_{+}^{m\times n}$ is a forward stoichiometric
matrix, $R\in\mathbb{Z}_{+}^{m\times n}$ is a reverse stoichiometric
matrix, with $\mathcal{F}$ and $\mathcal{R}$ being two sequences
of cardinality $n$. The entry $F_{i,j}$ is the stoichiometric number
of molecular species $i$ consumed in the $j^{th}$ directed reaction,
and the entry $R_{i,j}$ is the stoichiometric number of molecular
species $i$ produced in the $j^{th}$ directed reaction. One may
then define a net stoichiometric matrix as $N\coloneqq R-F\in\mathbb{Z}^{m\times n}$.
Note that the definition of a net stoichiometric matrix in terms of
forward and reverse stoichiometric matrices allows for a molecular
species, for example, an enzyme catalyst, to be both consumed and
produced in a reaction, in which case the corresponding net stoichiometric
coefficient is zero. However, henceforth, we do not consider a catalyst
in reactions. Note that the sequence of vertices and hyperedges is
arbitrary, but once these sequences are defined, they must be used
consistently in different theoretical representations.

\subsubsection{\label{subsec:Example-directed-stoichiometric}Example directed stoichiometric
hypergraph}

\begin{figure}
\begin{centering}
\includegraphics[width=8cm]{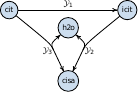}
\par\end{centering}
\caption{\label{fig:mitocell}\textbf{A directed stoichiometric hypergraph.
}The four molecular species (vertices) are citrate (\protect\href{https://vmh.life/\#reaction/cit}{cit},
$C_{6}H_{5}O_{7}$), isocitrate (\protect\href{https://www.vmh.life/\#metabolite/icit}{icit},
$C_{6}H_{5}O_{7}$) , cis-aconitate (\protect\href{https://vmh.life/\#metabolite/HC00342}{cisa},
$C_{6}H_{3}O_{6}$) and water (\protect\href{https://vmh.life/\#metabolite/h2o}{h2o},
$H_{2}O$). In biochemical terms, the reactions (black hyperedges)
are $\mathcal{Y}_{1}$: aconitate hydratase (\protect\href{https://www.vmh.life/\#reaction/ACONTm}{ACONTm}),
$\mathcal{Y}_{2}$: citrate hydro-lyase (\protect\href{https://vmh.life/\#reaction/r0317}{r0317})
and $\mathcal{Y}_{3}$: isocitrate hydro-lyase (\protect\href{https://vmh.life/\#reaction/r0426}{r0426}).
Although each reaction is, in principle, reversible, the directions
of each hyperedge are given in the conventional orientation, consistent
with the corresponding stoichiometric matrix. Figure taken from \cite{ghaderi_structural_2020}.}
\end{figure}

We consider the network defined in \cite{ghaderi_structural_2020}.
It represents a directed stoichiometric hypergraph with 4 molecular
species \textit{$\mathcal{V}=(\ensuremath{cit},\mathit{{\rm \mathit{ici}t},{\rm \mathit{cisa}},{\rm \mathit{h2o}}})$}
and 3 reactions $\mathcal{Y}=(\mathcal{Y}_{1},\mathcal{Y}_{2},\mathcal{Y}_{3})$,
a planar representation of which is illustrated in Figure \ref{fig:mitocell}.
The $3$ reaction equations are

\begin{eqnarray}
\mathcal{Y}_{1}:\;\,cit & \rightleftharpoons & icit,\nonumber \\
\mathcal{Y}_{2}:\;\,icit & \rightleftharpoons & h2o+cisa,\nonumber \\
\mathcal{Y}_{3}:\;cit & \rightleftharpoons & h2o+cisa.\label{eq:mito3rxn}
\end{eqnarray}
The corresponding net stoichiometric matrix is

\vspace{-2mm}

\begin{table}[H]
\centering{}%
\begin{tabular}{l|ccccc|l}
\multicolumn{1}{l}{} &  & $\mathcal{Y}_{1}$ & $\mathcal{Y}_{2}$ & $\mathcal{Y}_{3}$ & \multicolumn{1}{c}{} & \tabularnewline
\cline{2-2}\cline{6-6}
 &  & 0 & 1 & 1 &  & $h2o$\tabularnewline
$N:=$ &  & -1 & 0 & -1 &  & $cit$\tabularnewline
 &  & 1 & -1 & 0 &  & $icit$\tabularnewline
 &  & 0 & 1 & 1 &  & $cisa$\tabularnewline
\cline{2-2}\cline{6-6}
\end{tabular}
\end{table}

\vspace{-2mm}
where rows and columns correspond to molecular species and reactions,
respectively. 

It is important to note that $\mathcal{Y}_{1}$ is a lumped representation
of $\mathcal{Y}_{2}$ and $\mathcal{Y}_{3}$. Specifically, $\mathcal{Y}_{1}$
represents the overall reaction, while $\mathcal{Y}_{2}$ and $\mathcal{Y}_{3}$
describe the elementary steps of this process. 

\smallskip{}

\subsection{\label{par:Molecular-graph}Molecular species graph}

Consider a molecular species $\mathcal{V}_{k}\in\mathcal{V}$, its
\textit{molecular species graph} is a connected graph $\mathcal{G}_{k}=\mathcal{G}(\mathcal{X},\mathcal{\mathcal{B}},\mathcal{V}_{k})$
where each vertex $\mathcal{X}_{i}$ is an atom and each edge $\mathcal{B}_{j}$
is a chemical bond. A chemical bond $\mathcal{B}_{ij}$ between two
atoms, $\mathcal{X}_{i}$ and $\mathcal{X}_{j}$ is an undirected
edge between two atoms in a molecular graph

\[
\mathcal{B}_{ij}:=\{\mathcal{X}_{i},\mathcal{X}_{j}\}.
\]
That is, we do not consider interactions between more than two atoms,
as may occur with hydrogen bonding or Van der Waals forces. We assume
that a molecular graph represents the topology but not the three dimensional
geometry of a molecular species, so stereoisomers have the same molecular
graph. Let $p\coloneqq|\mathcal{X}(\mathcal{V}_{k})|$ denote the
cardinality of atoms of molecular species $\mathcal{V}_{k}$ and $q\coloneqq|\mathcal{B}(\mathcal{V}_{k})|$
denote the cardinality of bonds of molecular species $\mathcal{V}_{k}$. 

Let $\mathcal{G}(\mathcal{X},\mathcal{\mathcal{B}},\mathcal{V})$
be a graph composed of $m$ molecular graphs, where each molecular
graph $\mathcal{G}(\mathcal{X},\mathcal{\mathcal{B}},\mathcal{V}_{k})$is
a connected component. Each vertex in $\mathcal{G}(\mathcal{X},\mathcal{\mathcal{B}},\mathcal{V})$
is triply labelled, with $(i)$ an element label, which is a type
of chemical element $(ii)$ an atomic label $i\in1\ldots n(\mathcal{V})$,
which uniquely identifies each of the $p(\mathcal{V}_{k})$ atoms
in $\mathcal{V}_{k}$, and $(iii)$ a molecular label, which uniquely
identifies a molecular species. Each edge is doubly labelled, with
the two vertex labels that form the chemical bond. A \textit{molecular
graph} is a graph representation of a molecular species that shows
the atoms in the molecular species and the bonds between them. A molecular
graph provides information about the connectivity of the atoms in
the molecular species, as well as the number and type of bonds between
them.

\subsubsection{\label{subsec:Matrix-representation-of}Matrix representation of
a molecular graph}

\textit{A} molecular graph $\mathcal{G}_{k}$ can be represented by
an incidence matrix $B\in\mathbb{Z}^{p\times q}$ given by

\[
B_{i,j}\coloneqq\begin{cases}
-1 & \mathcal{X}_{i}\in\text{tail of edge}e_{j},\\
1 & \mathcal{X}_{i}\in\text{head of edge}e_{j},\\
0 & \textrm{otherwise,}
\end{cases}
\]
and a weight vector $w\in\mathbb{N}^{q\times1}$ given by

\[
w_{j}=k,
\]
where $k$ is a non-negative integer indicating the order of the $j$th
bond (0 for a non-existent bond, 1 for a single bond, 2 for a double
bond, and 3 for a triple bond). The rows of the incidence matrix $B$
correspond to the $p$ atoms of the molecular graph $\mathcal{G}_{k}$,
and its columns correspond to the $q$ bonds of the molecular graph
$\mathcal{G}_{k}$. 

\subsubsection{Example molecular graph}

An example of a matrix representation of a molecular graph is provided
for acetate (\href{https://www.vmh.life/\#metabolite/ac}{ac}, ac)
in Fig. \ref{fig:ac}.

\begin{figure}
\vspace{-2mm}

\begin{minipage}[b][1\totalheight][t]{0.33\columnwidth}%
\subfloat[\label{fig:ac}Molecular graph]{

\includegraphics[width=5cm]{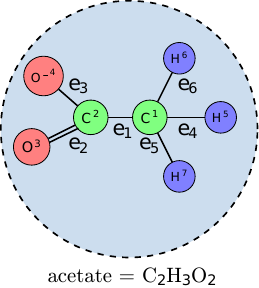}}%
\end{minipage}\hspace{-1cm}%
\begin{minipage}[b][1\totalheight][t]{0.33\columnwidth}%
\subfloat[\label{fig:acBEmatrix-1}Molecular incidence matrix and the weight
vector]{
\centering{}%
\begin{tabular}{l|cccccccc|l}
\multicolumn{1}{l}{} &  & $e_{1}$ & $e_{2}$ & $e_{3}$ & $e_{4}$ & $e_{5}$ & $e_{6}$ & \multicolumn{1}{c}{} & \tabularnewline
\cline{2-2}\cline{9-9}
 &  & -1 & 0 & 0 & -1 & -1 & -1 &  & $C^{1}$\tabularnewline
 &  & 1 & -1 & -1 & 0 & 0 & 0 &  & $C^{2}$\tabularnewline
 &  & 0 & 1 & 0 & 0 & 0 & 0 &  & $O^{3}$\tabularnewline
$B:=$ &  & 0 & 0 & 1 & 0 & 0 & 0 &  & $O^{-4}$\tabularnewline
 &  & 0 & 0 & 0 & 1 & 0 & 0 &  & $H^{5}$\tabularnewline
 &  & 0 & 0 & 0 & 0 & 0 & 1 &  & $H^{6}$\tabularnewline
 &  & 0 & 0 & 0 & 0 & 1 & 0 &  & $H^{7}$\tabularnewline
\cline{2-2}\cline{9-9}
\end{tabular}%
\begin{tabular}{l|ccc|l}
\multicolumn{1}{l}{} &  &  & \multicolumn{1}{c}{} & \tabularnewline
\cline{2-2}\cline{4-4}
 &  & 1 &  & $e_{1}$\tabularnewline
 &  & 2 &  & $e_{2}$\tabularnewline
 &  & 1 &  & $e_{3}$\tabularnewline
$w:=$ &  & 1 &  & $e_{4}$\tabularnewline
 &  & 1 &  & $e_{5}$\tabularnewline
 &  & 1 &  & $e_{6}$\tabularnewline
\cline{2-2}\cline{4-4}
\end{tabular}}%
\end{minipage}\negthinspace{}

\caption{\textbf{Acetate represented as a molecular graph and molecular incidence
matrix. }(a) The molecular graph of an acetate species (\protect\href{https://www.vmh.life/\#metabolite/ac}{ac},
ac) with the chemical formula $C_{2}H_{3}O_{2}$. The two types of
bonds are illustrated; single ($-$) and double ($=$). (b) The molecular
incidence matrix $B$ corresponds to the acetate species. Each row
corresponds to an atom, and each column corresponds to a chemical
bond in the molecular graph. Each column has two non-zero entries
where the rows correspond to the atoms forming the chemical bond.
The results are independent of the particular orientation chosen.
The weight vector $w$ is a vector where the entries represent the
type of bonds; single (1) and double (2).}
\end{figure}

\vspace{2mm}

\subsection{\label{subsec:Chemical-complex-graphs}Chemical complex graph}

Given a set of molecular species $\mathcal{V}$, a chemical complex
$\mathcal{C}(\mathcal{V})$ is a subset of molecular species that
participate together either as substrates, or products, in a reaction.
A complex graph $\mathcal{G}(\mathcal{X},\mathcal{\mathcal{\mathcal{B}}},\mathcal{C}(\mathcal{V}))$
is the disjoint union of a multiset of $\left|\mathcal{C}\right|$
molecular graphs, where each molecular graph corresponds to a molecular
species $\mathcal{V}_{k}\in\mathcal{C}$. Each vertex is triply labelled
with a molecular, elemental, and atomic labels. The total number of
vertices in complex graph $\mathcal{C}$ is 
\[
p\coloneqq\sum_{\mathcal{V}_{k}\in\mathcal{C}}|\mathcal{X}(\mathcal{V}_{k})|,
\]
where $|\mathcal{X}(\mathcal{V}_{k})|$ is the number of atoms in
molecular species $k$. The total number of edges in a chemical complex
$C(\mathcal{V})$ is
\begin{eqnarray*}
q & \coloneqq & \sum_{\mathcal{V}_{k}\in\mathcal{C}}|\mathcal{\mathcal{B}}(\mathcal{V}_{k})|,
\end{eqnarray*}
where $|\mathcal{\mathcal{B}}(\mathcal{V}_{k})|$ is the number of
bonds in molecular species $k$. The number of connected components
of a complex graph is equal to the molecularity of that complex. For
example, if a complex consists of two instances of the same molecular
species, then the complex graph contain two connected components that
are isomorphic up to vertex labelling, corresponding to a complex
with stoichiometric number (multiplicity) two for that molecular species.
A substrate chemical complex $\mathcal{S}(\mathcal{\mathcal{V}})$,
is a chemical complex formed by substrate molecular species instances
and a product chemical complex $\mathcal{P}(\mathcal{V})$ is a chemical
complex formed by product molecular species instances. Substrate and
product chemical complexes are related in pairs, one corresponding
to each reaction $\mathcal{Y}_{j}\coloneqq\{\mathcal{S}_{j}(\mathcal{\mathcal{V}}),\mathcal{P}_{j}(\mathcal{V})\}$.
A chemical complex matrix $B\in\{-1,0,1\}^{p\times q}$is an incidence
matrix representing a chemical complex $\mathcal{C}(\mathcal{V})$,
consisting of $m$ molecular graph incidence matrices arranged in
block diagonal form, where each block represents an instance of a
molecular species involved in that complex. 

\subsubsection{Example chemical complex graph}

Figure \ref{fig:achemicalComplex} represents an example of a chemical
complex matrix representing a complex of peroxynitrite (\href{https://www.vmh.life/\#metabolite/CE5643}{peroxynitrite})
and carbon dioxide.

\begin{figure}
\caption{\label{fig:achemicalComplex}Molecular incidence matrix for chemical
complex (peroxynitrite, carbon dioxide). The matrix $B$ represents
the molecular incidence matrix for the chemical complex (peroxynitrite,
carbon dioxide). The first block corresponds to the molecular incidence
matrix of peroxynitrite, and the second block corresponds to the molecular
incidence matrix of carbon dioxide.}

\centering{}%
\begin{tabular}{l|cccc|ccc|l}
\multicolumn{1}{l}{} &  & $e_{1}$ & $e_{2}$ & \multicolumn{1}{c}{$e_{3}$} & $e_{4}$ & $e_{5}$ & \multicolumn{1}{c}{} & \tabularnewline
\cline{2-2}\cline{8-8}
 &  & -1 & 0 & 0 & 0 & 0 &  & $O^{1}$\tabularnewline
 &  & 1 & -1 & 0 & 0 & 0 &  & $N^{2}$\tabularnewline
 &  & 0 & 1 & -1 & 0 & 0 &  & $O^{3}$\tabularnewline
$B:=$ &  & 0 & 0 & 1 & 0 & 0 &  & $O^{-4}$\tabularnewline
\cline{3-7}
 &  & 0 & 0 & 0 & -1 & 0 &  & $O^{5}$\tabularnewline
 &  & 0 & 0 & 0 & 1 & -1 &  & $C^{6}$\tabularnewline
 &  & 0 & 0 & 0 & 0 & 1 &  & $O^{7}$\tabularnewline
\cline{2-2}\cline{8-8}
\end{tabular}
\end{figure}

\subsection{\label{subsec:reactionMatrix}Reaction matrix}

A \textit{substrate matrix} is a chemical complex matrix that represents
the chemical complex formed by each instance of a substrate molecular
species. A \textit{product matrix} is a chemical complex matrix that
represents the chemical complex formed by each instance of a product
molecular species. Consider a reaction $\mathcal{Y}:=\{\mathcal{S}(\mathcal{V}),\mathcal{P}(\mathcal{V})\}$,
between a substrate complex $\mathcal{S}(\mathcal{V})$ and a product
complex $\mathcal{P}(\mathcal{V})$. Let $S\in\mathbb{Z}^{p\times\max(q,q')}$
be a substrate matrix and $P\in\mathbb{Z}^{p\times\max(q,q')}$ be
a product matrix, where $p$ is the number of atoms in the substrate
(or product) complex, $q$ is the number of bonds in the substrate
complex $\mathcal{S}(\mathcal{V})$ and $q'$ is the number of bonds
in the product complex $\mathcal{P}(\mathcal{V})$. Both substrate
and product complexes have the same number of atoms, so the number
of rows in the substrate and product complex matrices are the same
and we require that atom transitions are between atoms with the same
row indices in substrate and product complex matrices. Depending on
the number of bonds in the substrate complex, the number of bonds
in the product complex and the correspondence between these bonds,
the matrices $S$ or $P$ may contain additional zero columns in order
to ensure they have the same number of columns, but conserved bonds
must correspond to same column index in both matrices. Let $w_{s}\in\mathbb{N}^{\textrm{max}(q,q')\times1}$
denote the weight vector specifying the order of the bonds in the
substrate and $w_{p}\in\mathbb{N}^{\textrm{max}(q,q')\times1}$ the
weight vector specifying the order of the bonds in the product.

A chemical reaction $\mathcal{\mathcal{Y}}$ may represented by the
equation

\begin{equation}
D\coloneqq\lvert P\cdot\textrm{diag}(w_{p})\rvert-\rvert S\cdot\textrm{diag}(w_{s})\rvert\label{eq:-1-1-1}
\end{equation}
where $D\in\mathbb{Z}^{p\times\textrm{max}(q,q')}$ is a\textit{ reaction
matrix, which} is an incidence matrix where each row represents an
atom and each column represents a bond involved in the reaction. If
$D_{i,j}=0$, then the bond $j$ involving atom $i$ in the substrate
and product complex is conserved by the reaction. If $D_{i,j}$ is
negative, then atom $i$ in the substrate complex participates in
a bond $j$ that is broken during the reaction, while if $D_{i,j}$
is positive, then atom $i$ in the product complex participates in
a bond $j$ that is formed during the reaction. In a reaction, a \textit{reacting
bond} is a chemical bond that is broken, formed, or changes its order.
In a reaction, a \textit{conserved} bond is a chemical bond that is
not a reacting bond. 

\subsubsection{Example reaction matrix}

Consider the reaction illustrated in Figure \ref{fig:reactionEx},
where the substrate complexes are peroxynitrite (\href{https://www.vmh.life/\#metabolite/CE5643}{peroxynitrite})
and carbon dioxide and the product complex is the nitrosooxy carbonate
(\href{https://www.vmh.life/\#metabolite/CE6000}{nit}). The substrate
incidence matrix $S$ corresponding to the substrate complex, and
the product incidence matrix $P$ corresponding to the product complex,
are given in Figure \ref{fig:Substrate and product matrices}. In
the figure \ref{fig:R}, the matrix $D$ is the reaction incidence
matrix representing the reacting bonds in the reaction \ref{fig:reactionEx}.

\begin{figure}
\selectlanguage{american}%
\,\,\,\,\,\,\,\,\,\,\,\,\,\subfloat[\foreignlanguage{british}{\label{fig:reactionEx}The chemical conversion of \foreignlanguage{british}{peroxynitrite
(\protect\href{https://www.vmh.life/\#metabolite/CE5643}{peroxynitrite})
and $CO_{2}$ into nitrosooxy carbonate (\protect\href{https://www.vmh.life/\#metabolite/CE6000}{nit}).}}]{\selectlanguage{british}%
\includegraphics[width=15cm]{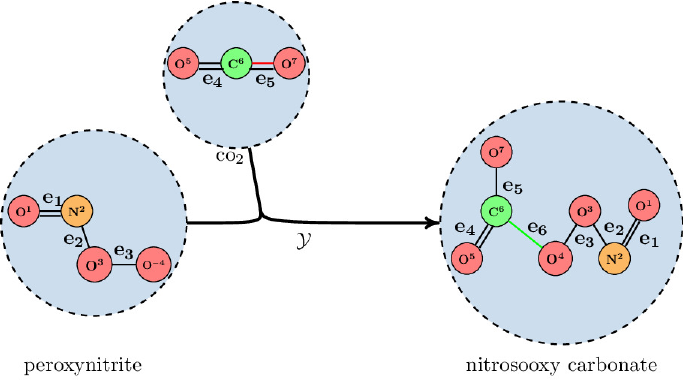}

\selectlanguage{american}%
}

\selectlanguage{british}%
\vspace{-2mm}

\begin{centering}
\hspace{-1.7cm}\subfloat[Substrate and product incidence matrices.]{\begin{centering}
\begin{tabular}{l|cccccccc|l}
\multicolumn{1}{l}{} &  & $e_{1}$ & $e_{2}$ & $e_{3}$ & $e_{4}$ & $e_{5}$ & $e_{6}'$ & \multicolumn{1}{c}{} & \tabularnewline
\cline{2-2}\cline{9-9}
 &  & -2 & 0 & 0 & 0 & 0 & 0 &  & $O^{1}$\tabularnewline
 &  & 2 & -1 & 0 & 0 & 0 & 0 &  & $N^{2}$\tabularnewline
 &  & 0 & 1 & -1 & 0 & 0 & 0 &  & $O^{3}$\tabularnewline
$S\cdot\textrm{diag}(w_{s}):=$ &  & 0 & 0 & 1 & 0 & 0 & 0 &  & $O^{4}$\tabularnewline
 &  & 0 & 0 & 0 & -2 & 0 & 0 &  & $O^{5}$\tabularnewline
 &  & 0 & 0 & 0 & 2 & -2 & 0 &  & $C^{6}$\tabularnewline
 &  & 0 & 0 & 0 & 0 & 2 & 0 &  & $O^{7}$\tabularnewline
\cline{2-2}\cline{9-9}
\end{tabular}%
\begin{tabular}{l|cccccccc|l}
\multicolumn{1}{l}{} &  & $e_{1}$ & $e_{2}$ & $e_{3}$ & $e_{4}$ & $e_{5}$ & $e_{6}$ & \multicolumn{1}{c}{} & \tabularnewline
\cline{2-2}\cline{9-9}
 &  & -2 & 0 & 0 & 0 & 0 & 0 &  & $O^{1}$\tabularnewline
 &  & 2 & -1 & 0 & 0 & 0 & 0 &  & $N^{2}$\tabularnewline
 &  & 0 & 1 & -1 & 0 & 0 & 0 &  & $O^{3}$\tabularnewline
$P\cdot\textrm{diag}(w_{p}):=$ &  & 0 & 0 & 1 & 0 & 0 & -1 &  & $O^{4}$\tabularnewline
 &  & 0 & 0 & 0 & -2 & 0 & 0 &  & $O^{5}$\tabularnewline
 &  & 0 & 0 & 0 & 2 & -1 & 1 &  & $C^{6}$\tabularnewline
 &  & 0 & 0 & 0 & 0 & 1 & 0 &  & $O^{7}$\tabularnewline
\cline{2-2}\cline{9-9}
\end{tabular}
\par\end{centering}
}
\par\end{centering}
\vspace{2mm}

\begin{centering}
\subfloat[\foreignlanguage{british}{\label{fig:R} Reaction incidence matrix.}]{\begin{centering}
\begin{tabular}{l|cccccccc|l}
\multicolumn{1}{l}{} &  & $h_{1}$ & $h_{2}$ & $h_{3}$ & $h_{4}$ & $h_{5}$ & $h_{6}$ & \multicolumn{1}{c}{} & \tabularnewline
\cline{2-2}\cline{9-9}
 &  & 0 & 0 & 0 & 0 & 0 & 0 &  & $O^{1}$\tabularnewline
 &  & 0 & 0 & 0 & 0 & 0 & 0 &  & $N^{2}$\tabularnewline
 &  & 0 & 0 & 0 & 0 & 0 & 0 &  & $O^{3}$\tabularnewline
$D:=$ &  & 0 & 0 & 0 & 0 & 0 & 1 &  & $O^{4}$\tabularnewline
 &  & 0 & 0 & 0 & 0 & 0 & 0 &  & $O^{5}$\tabularnewline
 &  & 0 & 0 & 0 & 0 & -1 & 1 &  & $C^{6}$\tabularnewline
 &  & 0 & 0 & 0 & 0 & -1 & 0 &  & $O^{7}$\tabularnewline
\cline{2-2}\cline{9-9}
\end{tabular}
\par\end{centering}
}
\par\end{centering}
\caption{\label{fig:Substrate and product matrices}\foreignlanguage{british}{The
chemical conversion of }peroxynitrite (\protect\href{https://www.vmh.life/\#metabolite/CE5643}{peroxynitrite})
and $CO_{2}$ into nitrosooxy carbonate (\protect\href{https://www.vmh.life/\#metabolite/CE6000}{nit}).
(a) The chemical bonds are represented by the edges in the molecular
graphs. In the reaction, a double bond ($C^{6}-O^{7}$, $e_{5}$)
is broken in $CO_{2}$, and a bond ($C^{6}-O^{4}$, $e_{6}$) is formed.
The chemical bonds in the substrate complex corresponding to those
in the product complex have the same labelling $(e_{1},e_{2},e_{3},e_{4})$,
which represent the conserved bonds. (b) $S$ is the substrate incidence
matrix corresponding to the peroxynitrite and $CO_{2}$. $P$ is the
product incidence matrix corresponding to nitrosooxy carbonate. (c)
$D$ is the incidence reaction matrix calculated by Eq \ref{eq:-1-1-1},
it represents the reacting bonds in the reaction. The matrix $D$
has two negative values between atom $C^{6}$ and atom $O^{7}$, representing
the broken double bond ($C^{6}-O^{7}$, $h_{5}$), and one positive
value between atom $C^{4}$and atom $C^{6}$ representing the formed
bond ($C^{6}-O^{4}$, $h_{6}$). In this reaction, the conserved atoms
are $\{O^{1},N^{2},O^{3},O^{5}\}$ while the reacting atoms are $\{O^{4},C^{6},O^{7}\}$.}
\end{figure}

\vspace{2mm}

\subsection{\label{subsec:Atom-mapping}Atom mapping}

Given a substrate chemical complex $\mathcal{S}(\mathcal{\mathcal{V}})$,
a product chemical complex $\mathcal{P}(\mathcal{V})$ and a reaction
$\mathcal{Y}\coloneqq\{\mathcal{S}(\mathcal{\mathcal{V}}),\mathcal{P}(\mathcal{V})\}$,
an atom transition is a labelled edge $\mathcal{E}:=\{\mathcal{X}_{i},\,\mathcal{X}_{j}\}$
that joins vertex $\mathcal{X}_{i}$ of molecular species $\mathcal{\mathcal{V}}_{k}$
in complex graph $\mathcal{G}(\mathcal{X},\mathcal{\mathcal{Y}},\mathcal{S}(\mathcal{\mathcal{V}}))$
with vertex $\mathcal{X}_{j}$ of molecular species $\mathcal{\mathcal{V}}_{l}$
in complex graph $\mathcal{G}(\mathcal{X},\mathcal{\mathcal{Y}},\mathcal{P}(\mathcal{V}))$.
The edge is labelled with a reaction label, which associates it with
a unique reaction. The element label of the vertex $\mathcal{X}_{i}\in\mathcal{G}(\mathcal{X},\mathcal{\mathcal{Y}},\mathcal{S})$
is the same as the element label of the vertex $\mathcal{X}_{j}\in\mathcal{G}(\mathcal{X},\mathcal{\mathcal{Y}},\mathcal{P})$.
That is, an atom transition is an edge between a pair of atoms of
the same element, one in each of the pair of complexes involved in
a reaction. Therefore, in a reaction, the total number of atoms of
each element in both complexes is the same. The molecular and atomic
labels may be different for both vertices in an atom mapping.

Given a set of molecular species $\mathcal{V}$ and a reaction $\mathcal{Y}\coloneqq\{\mathcal{S}(\mathcal{\mathcal{V}}),\mathcal{P}(\mathcal{V})\}$,
an atom mapping is a graph $\mathcal{G}(\mathcal{X},\mathcal{\mathcal{E}},\mathcal{H}\{\mathcal{S}(\mathcal{\mathcal{V}}),\mathcal{P}(\mathcal{V})\})$
formed by the disjoint union of the set of
\[
\left|\mathcal{E}\right|\coloneqq\sum_{\mathcal{V}_{k}\in\mathcal{S}}|\mathcal{X}(\mathcal{V}_{k})|=\sum_{\mathcal{V}_{k}\in\mathcal{P}}|\mathcal{X}(\mathcal{V}_{k})|
\]
atom transitions, between 
\[
\left|\mathcal{X}\right|\coloneqq\sum_{\mathcal{V}_{k}\in\mathcal{S}}|\mathcal{X}(\mathcal{V}_{k})|+\sum_{\mathcal{V}_{k}\in\mathcal{P}}|\mathcal{X}(\mathcal{V}_{k})|=2\left|\mathcal{E}\right|
\]
vertices. Each edge is labelled with an identical reaction label.
Each vertex is labelled with an element label, a molecular label and
an atomic label. Note that an atom mapping consists of $\left|\mathcal{E}\right|$
connected components, each of which contains one edge and two vertices
with identical element labels. That is, all edges of the molecular
graphs of each molecular species in $\mathcal{V}$ are omitted. One
reaction may correspond to multiple alternative atom mappings, e.g.,
if a molecular topology has a symmetrical subgraph, this may permit
multiple alternate atom mappings whose vertices are equivalent with
respect to element labelling, but not with respect to atomic labelling.

\subsubsection{Example atom mapping}

Figure \foreignlanguage{british}{\ref{fig:reactionEx}} illustrates
an atom mapping for a chemical reaction, where the substrates are
peroxynitrite (\href{https://www.vmh.life/\#metabolite/CE5643}{peroxynitrite})
and $CO_{2}$, and the product complex is nitrosooxycarbonate (\href{https://www.vmh.life/\#metabolite/CE6000}{nit}).

\subsection{\label{subsec:Generation-of-ATM}Directed atom transition multigraph}

Given a directed stoichiometric hypergraph $\mathcal{H}(\mathcal{X},\mathcal{Y}\{\mathcal{S}(\mathcal{V}),\mathcal{P}(\mathcal{V})\})$
and an atom \\
mapping $\mathcal{G}(\mathcal{X},\mathcal{\mathcal{Y}},\mathcal{H}\{\mathcal{S}(\mathcal{V}),\mathcal{P}(\mathcal{V})\})$
for each reaction, a \textit{directed atom transition multigraph }$\mathcal{G}(\mathcal{X},\mathcal{\mathcal{E}},\mathcal{H})$\textit{
is a multigraph }formed by the union of a set of $n\coloneqq\left|\mathcal{Y}\right|$
directed atom mappings, each of which corresponds to a reaction. The
union merges vertices of atom mappings that have identical molecular,
elemental and atomic labels, but duplicates edges if they have the
same head and tail vertices. Each of the $p\coloneqq\left|\mathcal{X}\right|$
vertices corresponds to an atom of an element in one of the $m\coloneqq\left|\mathcal{V}\right|$
molecular species, so each vertex is labelled with molecular, elemental
and atomic labels. Each of the $t\coloneqq\left|\mathcal{E}\right|$
edges corresponds to a directed atom transition in an atom mapping
corresponding to one of the $n\coloneqq\left|\mathcal{Y}\right|$
reactions, so each edge is labelled with a reaction label. The topology
of \textit{a directed atom transition multigraph} is represented by
an incidence matrix $T\in\left\{ -1,0,1\right\} ^{p\times t}$, where
each row is an instance of a chemical element in a particular molecular
species, and each directed edge is a\textit{ directed atom transition}.

A stoichiometric matrix $N\in\mathbb{Z}^{m\times n}$ may be related
to the incidence matrix of the corresponding \textit{directed atom
transition multigraph }$T\in\left\{ -1,0,1\right\} ^{p\times t}$\textit{
by defining }two mapping matrices, as follows. Let $V\in\left\{ 0,1\right\} ^{m\times p}$
denote a matrix that maps each molecular species to each atom, that
is $V_{i,j}=1$ if molecular species $i$ contains atom j, and $V_{i,j}=0$
otherwise. Each column of $V$ contains a single $1$ since each atom
is labelled with molecular, and atomic labels and is therefore specific
to a particular molecular species. Let $E\in\left\{ 0,1\right\} ^{t\times n}$
denote a matrix that maps each directed atom transition to each reaction,
that is $E_{i,j}=1$ if directed atom transition $i$ occurs in reaction
$j$ and $E_{h,j}=0$ otherwise. Then a stoichiometric matrix $N$
can be decomposed in terms of its \textit{directed atom transition
multigraph} with
\begin{equation}
N=\left(VV^{\mathrm{T}}\right)^{-1}VTE.\label{eq:multigraphDecomposition}
\end{equation}
The decomposition in Eq. \ref{eq:multigraphDecomposition} can more
easily be interpreted by rearranging terms to obtain,
\begin{equation}
VV^{\mathrm{\mathrm{T}}}N=VTE.\label{eq:VME-1}
\end{equation}
Since each column of $V$ contains a single $1$, the matrix $VV^{\mathrm{T}}\in\mathbb{N}^{m\times m}$
is a diagonal matrix with the total number of atoms in each molecular
species along the diagonal. The right hand side of Eq. \ref{eq:VME-1}
is therefore the internal stoichiometric matrix with each row scaled
by the total number of atoms in the corresponding molecular species.
Every molecular species contains at least one atom, so $\left(VV^{\mathrm{T}}\right)$
is invertible.

\subsection{\textit{\label{subsec:Generation-of-ATG}Atom transition graph}}

Given a \textit{directed atom transition multigraph,} an \textit{atom
transition graph is an undirected graph }$\mathcal{T}(\mathcal{X},\mathcal{\mathcal{E}},\mathcal{H})$
formed by removing duplicate vertices, that have identical elemental
and atomic labels, and by removing edges that are identical when head
and tail vertices are swapped. Each of the $p\coloneqq\left|\mathcal{X}\right|$
vertices corresponds to an atom of an element in one of the $m\coloneqq\left|\mathcal{V}\right|$
molecular species and is labelled with molecular, elemental and atomic
labels. Each of the $q\coloneqq\left|\mathcal{E}\right|$ edges corresponds
to an atom transition in one or more atom mappings and is unlabelled. 

Let $T\in\left\{ -1,0,1\right\} ^{p\times q}$ denote the incidence
matrix of an atom transition graph. Let $E\in\left\{ -1,0,1\right\} ^{q\times n}$
denote a matrix that maps each atom transition to one or more reactions,
that is $E_{i,j}=1$ if atom transition $i$ occurs with the same
orientation as reaction $j$, $E_{i,j}=-1$ if atom transition $i$
occurs with the opposite orientation to reaction $j$ and $E_{h,j}=0$
otherwise. The internal stoichiometric matrix $N$ can be decomposed
in terms of \textit{an atom transition graph with}
\begin{equation}
N=\left(VV^{\mathrm{T}}\right)^{-1}VTE.\label{eq:graphDecomposition}
\end{equation}
Note that the dimension of the incidence matrices representing a directed
atom transition multigraph and a corresponding atom transition graph
may not be the same as the latter may have fewer columns, that is
$t\ge q$. Furthermore, for an atom transition graph, the matrix $E$
has entries in the set $\{-1,0,1\}$ rather than just $\{0,1\}$,
to reflect reorientation with respect to certain reactions. However,
the matrix $V$ is the same for the decomposition of a stoichiometric
matrix in terms of a directed atom transition multigraph or an atom
transition graph.

\subsection{\label{subsec:Molecular-transition-graph}Molecular transition graph}

Given a directed stoichiometric hypergraph, a molecular transition
graph is \textit{an undirected graph} that is the union of the corresponding
molecular graph $\mathcal{G}(\mathcal{X},\mathcal{\mathcal{B}},\mathcal{V})$
and atom transition graph $\mathcal{T}(\mathcal{X},\mathcal{E},\mathcal{H})$.
In a molecular transition graph, each vertex is an atom and each edge
either corresponds to a bond in a molecular species or to an atom
transition in one or more reactions. Accordingly, a molecular transition
graph is denoted $\mathcal{L}(\mathcal{X},\mathcal{B},\mathcal{E},\mathcal{H})$,
where $\mathcal{X}$ is the set of atoms, $\mathcal{B}$ is the set
of bonds, $\mathcal{E}$ is the set of atom transitions and $\mathcal{H}$
is the stoichiometric hypergraph $\mathcal{H}=\mathcal{H}(\mathcal{X},\mathcal{Y}\{\mathcal{S}(\mathcal{V}),\mathcal{P}(\mathcal{V})\})$.
Each vertex is labelled with molecular, elemental and atomic labels.
Each bond edge is doubly labelled, with the two vertex labels that
form the chemical bond, and each atom transition edge is unlabelled.
The topology of a molecular transition graph\textit{ }$\mathcal{L}(\mathcal{X},\mathcal{B},\mathcal{E},\mathcal{H})$
is given by an incidence matrix $A\in\left\{ -1,0,1\right\} ^{p\times q}$,
where $p=\left|\mathcal{X}\right|$ and $q=\left|\mathcal{B}\right|+\left|\mathcal{E}\right|$.

\subsubsection{Example of a molecular transition graph}

Figure \ref{fig:MolecularTransition_reactionEx} illustrates the molecular
transition graph for the chemical reaction shown in Figure \ref{fig:reactionEx}. 

\begin{figure}
\includegraphics[width=15cm]{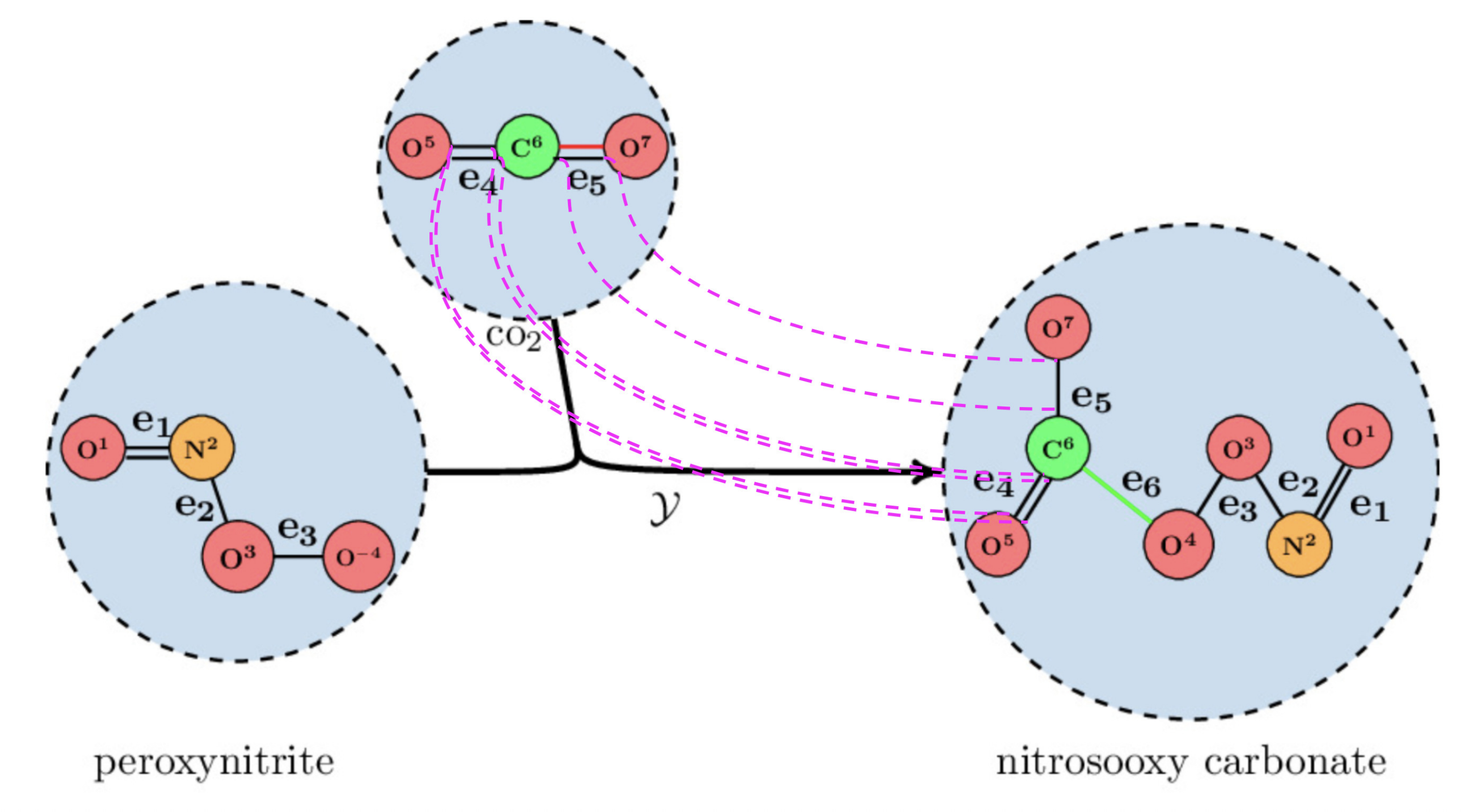}

\caption{\foreignlanguage{british}{\label{fig:MolecularTransition_reactionEx} \textbf{Molecular transition
graph. }For simplicity, only the atom transitions involving atoms
$O5$, C6 and O7 are shown. Each atom in the substrate is connected
to its corresponding atom in the product by a dashed magenta edge,
indicating the atom transition. The double bond $e_{4}$ between atoms
$O5$ and $C6$ is conserved, whereas the double bond $e_{5}$changes
order.}}
\end{figure}

\section{\label{subsec:Conserved-and-reacting-MTG}Conserved and reacting
graphs}

In a molecular transition graph, an edge corresponds to a reacting
bond if the bond is broken, formed or changes its order in at least
one reaction, otherwise it is a \textit{conserved} bond. In a molecular
transition graph an atom is an ambivorous atom if it participates
in at least one reacting bond, otherwise it is a \textit{conserved}
atom. These descriptive definitions enable a molecular transition
graph to be partitioned into conserved and reacting subgraphs. Next
these definitions are given in graph theoretical terms leading to
a partition of the molecular transition graph incidence matrix.

\subsection{Bond transition graphs}

Consider a bond $\mathcal{B}_{ij}:=\{\mathcal{X}_{i},\mathcal{X}_{j}\}$
in a molecular transition graph $\mathcal{L}(\mathcal{X},\mathcal{B},\mathcal{E},\mathcal{H})$.
Associated with vertex $\mathcal{X}_{i}$ and $\mathcal{X}_{j}$ are
two corresponding connected components, $\mathcal{C}_{i}$ and $\mathcal{C}_{j}$,
of the atom transition graph $\mathcal{G}(\mathcal{X},\mathcal{E},\mathcal{H})$,
which is, by definition, a subgraph of the molecular transition graph
$\mathcal{L}(\mathcal{X},\mathcal{B},\mathcal{E},\mathcal{H})$ that
includes all atom transition edges but no bond edges. Let $\mathcal{G}_{\mathcal{C}_{i},\mathcal{C}_{j}}$
denote the subgraph of the molecular transition graph $\mathcal{G}(\mathcal{X},\mathcal{\mathcal{B}},\mathcal{V})$
representing all of the bonds connecting at one atom of $\mathcal{C}_{i}$
with one atom of $\mathcal{C}_{j}$, that is 
\[
\mathcal{G}_{\mathcal{C}_{i},\mathcal{C}_{j}}\coloneqq\{\mathcal{B}_{ij}\text{\ensuremath{\mid}}\mathcal{X}_{i}\in\mathcal{C}_{i},\mathcal{X}_{j}\in\mathcal{C}_{j},\mathcal{B}_{ij}\in\mathcal{B}\}.
\]
The $k^{th}$ bond transition graph $\mathcal{L}_{k}(\mathcal{C}_{j},\mathcal{C}_{i})$,
of a molecular transition graph $\mathcal{L}(\mathcal{X},\mathcal{B},\mathcal{E},\mathcal{H})$,
is the union of the connected components $\mathcal{C}_{i}$ and $\mathcal{C}_{j}$
with $\mathcal{G}_{\mathcal{C}_{i},\mathcal{C}_{j}}$, that is
\[
\mathcal{L}_{k}\coloneqq\mathcal{C}_{j}\cup\mathcal{C}_{i}\cup\mathcal{G}_{\mathcal{C}_{i},\mathcal{C}_{j}}.
\]
where $\mathcal{C}_{i}$ and $\mathcal{C}_{j}$ are two connected
components in the atom transition graph $\mathcal{G}(\mathcal{X},\mathcal{E},\mathcal{H})$,
and $\mathcal{G}_{\mathcal{C}_{i},\mathcal{C}_{j}}$is the subgraph
of the molecular graph $\mathcal{G}(\mathcal{X},\mathcal{\mathcal{B}},\mathcal{V})$
that represents the bonds connecting atoms in $\mathcal{C}_{i}$ and
$\mathcal{C}_{j}$. When no bond exists between any pair of atoms
in $\mathcal{C}_{i}$ and $\mathcal{C}_{j}$, then $\mathcal{G}_{\mathcal{C}_{i},\mathcal{C}_{j}}$
is an empty set. Note that, a double bond, corresponds to two bond
transition graphs, one for each bond individual bonding interaction,
to enable consideration of reactions where a double bond is replaced
by a single bond, or vice versa. 

A bond $\mathcal{B}_{ij}$ between atoms $\mathcal{X}_{i}$ and $\mathcal{X}_{j}$
is a conserved bond if it remains unchanged in all reactions 

\[
\forall\mathcal{Y}\in\mathcal{H},\text{ }\mathcal{B}_{ij}\in\mathcal{S}(\mathcal{Y})\cap\mathcal{P}(\mathcal{Y}),
\]
where $\mathcal{S}(\mathcal{Y})$ and $\mathcal{P}(\mathcal{Y})$
are the sets of bonds in the substrate and product complexes of reaction
$\mathcal{Y}$, respectively. A \textit{conserved bond transition
graph} $\bar{\mathcal{L}}_{k}$ is a bond transition graph whose bonds
are conserved (neither created, broken or changed order) in any reaction
in a molecular transition graph $\mathcal{L}(\mathcal{X},\mathcal{B},\mathcal{E},\mathcal{H})$.
A bond $\mathcal{B}_{ij}$ between atoms $\mathcal{X}_{i}$ and $\mathcal{X}_{j}$
is a reacting bond if it is either formed, broken or changes order
in a reaction 

\[
\forall\mathcal{Y}\in\mathcal{H},\text{ }\mathcal{B}_{ij}\in\mathcal{S}(\mathcal{Y})\,\triangle\,\mathcal{P}(\mathcal{Y}),
\]
where $\triangle$ represents the symmetric difference between the
substrate bond set $\mathcal{S}(\mathcal{Y})$ and the product bond
set $\mathcal{P}(\mathcal{Y})$, indicating that a bond is either
formed, broken or changes order during the reaction. A reacting\textit{
bond transition graph} $\hat{\mathcal{\mathcal{L}}}_{k}$ is a bond
transition graph where at least one reaction involves a reacting bond.
Thus, $\hat{\mathcal{\mathcal{L}}}$ captures the molecular transitions
that involve bond changes, presenting the reacting bonds in the chemical
network. Each bond transition graph is either a conserved or reacting\textit{
bond transition graph.} 

\subsubsection{Examples of conserved and reacting bond transition graphs}

Figure \ref{fig:ConservedReactingBondTransitionGraph} illustrates
an example of a conserved bond transition graphs and an example of
a reacting bond transition graph. 

\begin{figure}
\subfloat[]{\includegraphics[scale=0.25]{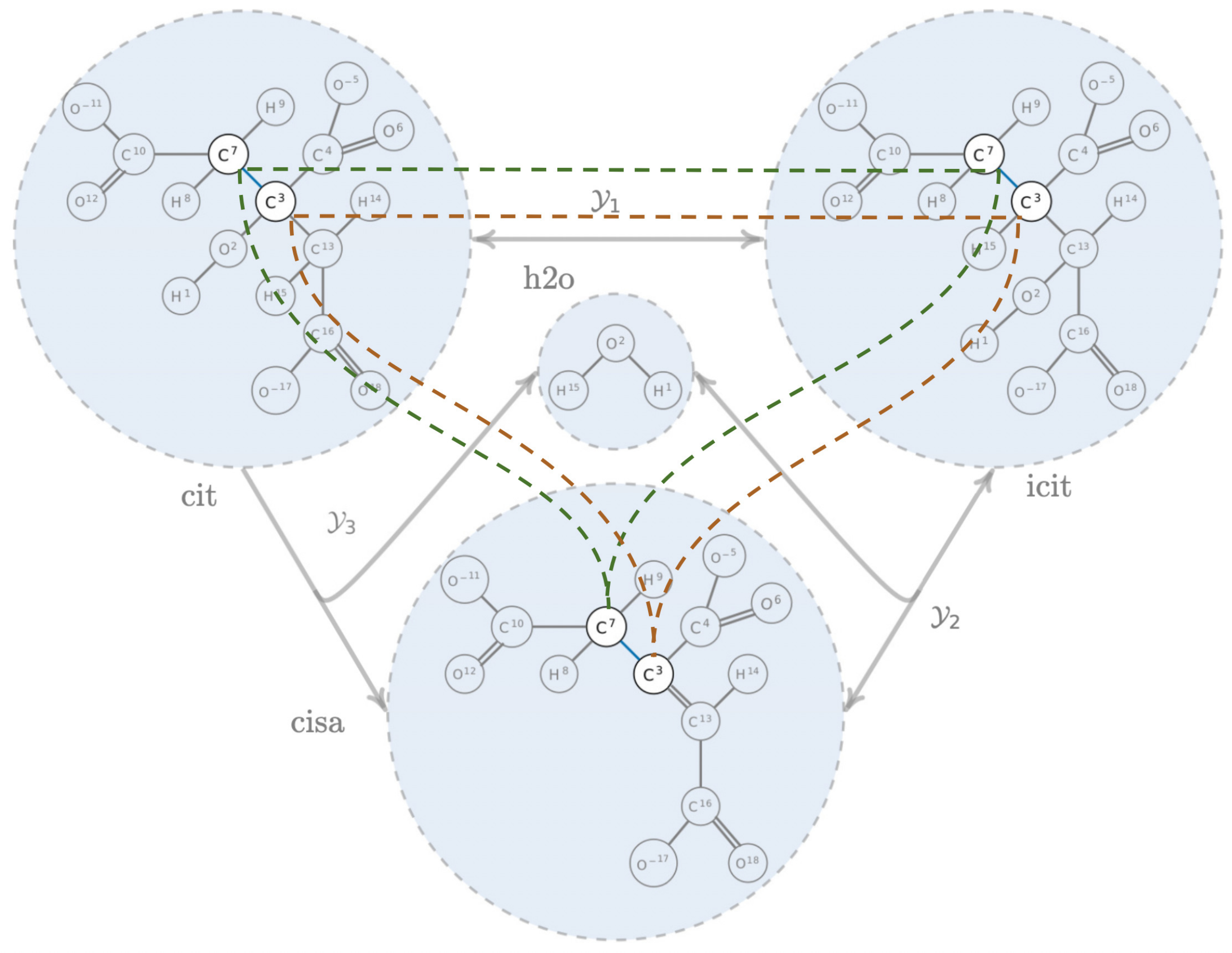}

}\subfloat[]{\includegraphics[scale=0.25]{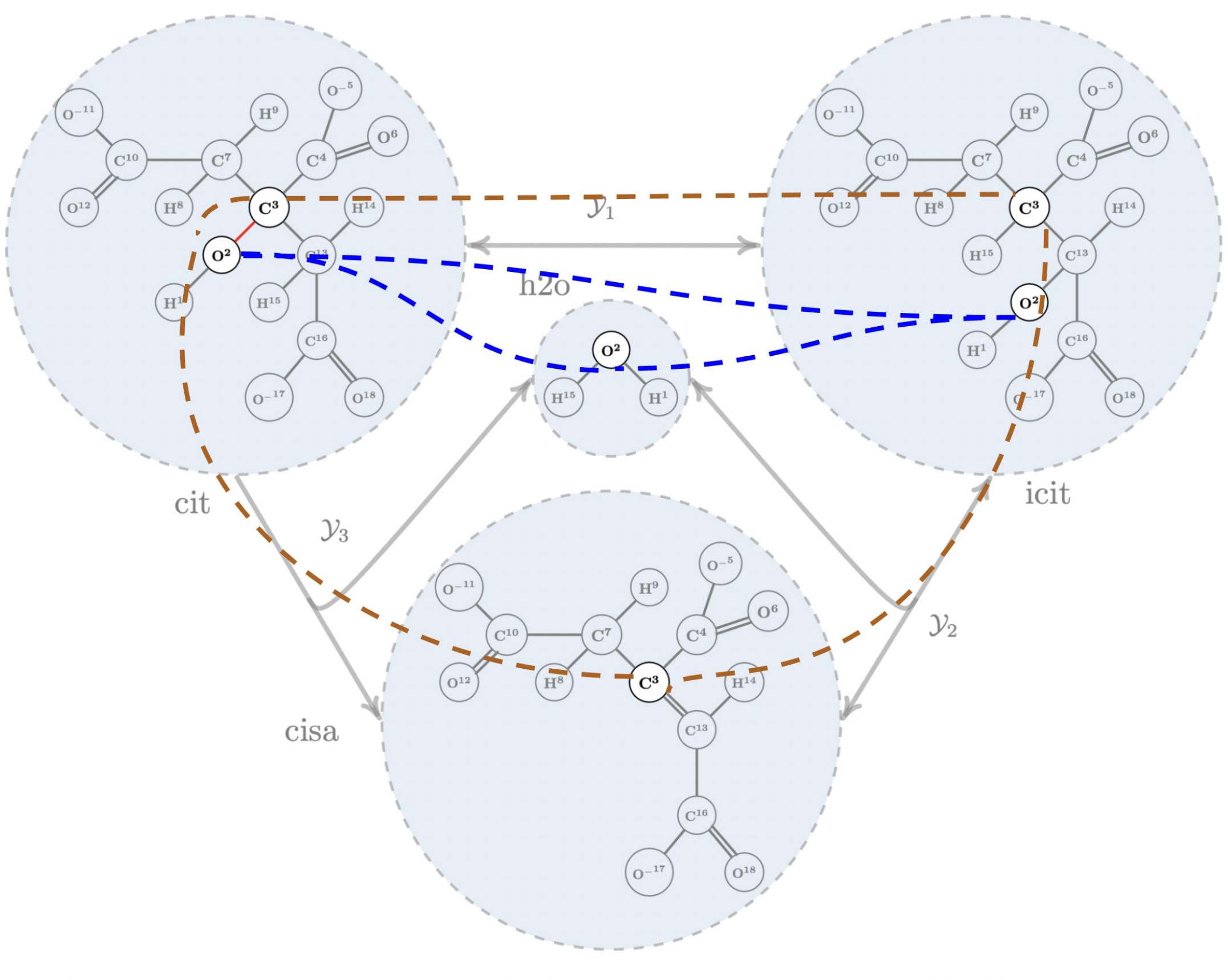}}\caption{\textbf{\label{fig:ConservedReactingBondTransitionGraph}Conserved
and reacting bond transition graph. }Atom transitions are dashed and
labelled with colours corresponding to components $C3$ (dark orange),
$C7$ (green) and $O2$ (dark blue). (a) \textbf{Conserved bond transition
graph.} The bond $C7-C3$ is conserved in reaction $\mathcal{Y}_{1},$$\mathcal{Y}_{2},$and
$\mathcal{Y}_{3}$ (continuous blue edges). (b) \textbf{Reacting bond
transition graph.} The bond $C3-O2$ is broken in reaction $\mathcal{Y}_{1}$
and in reaction $\mathcal{Y}_{3}$ (reacting bonds are in continuous
red lines).}
\end{figure}

\subsection{Conserved and reacting molecular transition graphs}

\textit{A conserved molecular transition graph} $\bar{\mathcal{L}}(\mathcal{X},\bar{\mathcal{B}},\mathcal{E},\mathcal{H})$
is the union of all conserved bond transition graphs of a molecular
transition graph, that is 
\[
\bar{\mathcal{L}}\coloneqq\bigcup_{k}\bar{\mathcal{L}}_{k}\subset\mathcal{L},
\]
where $\bar{\mathcal{B}}$ denotes a set of conserved bonds. A reacting
molecular transition graph $\hat{\mathcal{\mathcal{L}}}(\mathcal{X},\hat{\mathcal{B}},\mathcal{E},\mathcal{H})$
is the union of all reacting bond transition graphs of a molecular
transition graph, that is 
\[
\hat{\mathcal{\mathcal{L}}}\coloneqq\bigcup_{k}\hat{\mathcal{\mathcal{L}}}_{k}\subset\mathcal{L},
\]
where $\hat{\mathcal{B}}$ denotes a set of reacting bonds. Each vertex
in a reacting molecular transition graph is termed an ambivorous atom
as it is also a vertex in a \textit{conserved molecular transition
graph. A} molecular transition graph is the union of a conserved and
a reacting molecular transition graph, that is 
\[
\mathcal{L}=\bar{\mathcal{L}}\cup\hat{\mathcal{\mathcal{L}}}.
\]
A molecular transition graph\textit{ }$\mathcal{L}(\mathcal{X},\mathcal{B},\mathcal{E},\mathcal{H})$
is given by an incidence matrix $A\in\left\{ -1,0,1\right\} ^{\left|\mathcal{X}\right|\times\left(\left|\mathcal{B}\right|+\left|\mathcal{E}\right|\right)}$.
The columns of this incidence matrix may be partitioned into one subset
of edges corresponding to conserved bonds $\bar{B}\in\left\{ -1,0,1\right\} ^{\left|\mathcal{X}\right|\times\left|\bar{\mathcal{B}}\right|}$
and atom transitions $T\in\left\{ -1,0,1\right\} ^{\left|\mathcal{X}\right|\times\left|\mathcal{E}\right|}$
and one subset of edges corresponding to reacting bonds $\hat{B}\in\left\{ -1,0,1\right\} ^{\left|\mathcal{X}\right|\times\left|\hat{\mathcal{B}}\right|}$,
that is

\begin{equation}
A=\left[\begin{array}{ccc}
\bar{A} & \hat{A} & T\end{array}\right]\label{eq:partitionedMolecularGraph}
\end{equation}

\subsubsection{Example of a conserved and reacting molecular transition graphs}

Figure \ref{fig:Reacting-moieties-molecular} illustrates the distinction
between conserved and reacting graphs respect to the network introduced
in \ref{subsec:Example-directed-stoichiometric}.

\begin{figure}[H]
\begin{centering}
\includegraphics[scale=0.95]{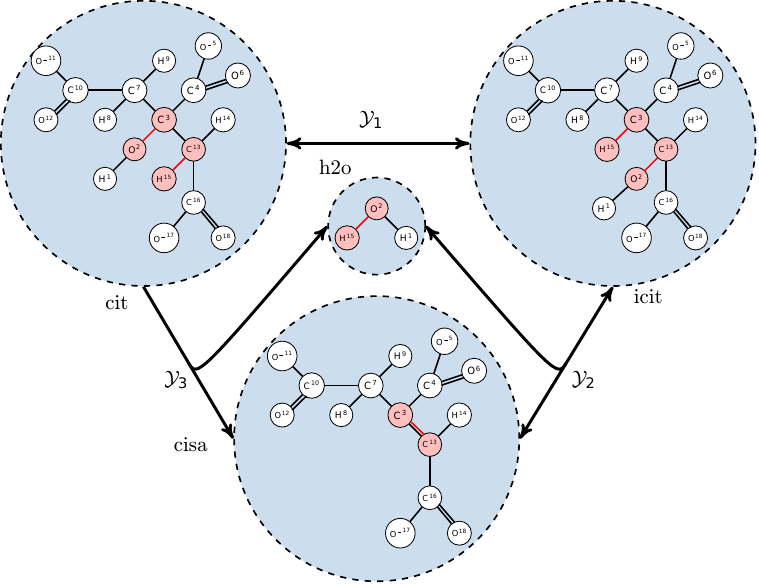}
\par\end{centering}
\caption{\textbf{\label{fig:Reacting-moieties-molecular}Molecular transition
graph partition.} Each bond is either conserved (black) or reacting
(red), while each atom is either conserved (white) or ambivorous (pink).
Note that both atoms involved in a reacting bond are reacting atoms
reacting atom, while a reacting atom may also be involved in a conserved
bond.}
\end{figure}

\section{\label{subsec:Identification-of-conserved}Conserved moieties}

In this section, we identify a conserved moiety as a species representing
a set of conserved moiety instances, with identical molecular topology
wherever they occur within the molecular graphs of a molecular network,
and are invariant with respect to all chemical transformations in
that network. First, we identify the set of atoms contained in each
conserved moiety instance by analysis of an atom transition graph,
then we identify the set of bonds contained in each conserved moiety
by analysis of the corresponding \textit{conserved molecular transition
graph.}

\subsection{\label{subsec:Identification-of-CATM}Connected components of an
atom transition graph}

Consider the connected components of an atom transition graph, $\mathcal{T}(\mathcal{X},\mathcal{\mathcal{E}},\mathcal{H})$.
Lemma \ref{thm:splittingN} demonstrates that the incidence matrix
representing a graph can be expressed as the sum of a set of incidence
matrices corresponding to its connected components. Let $C\in\{0,1\}^{c\times p}$
be a mapping between connected components and atoms in an atom transition
graph, where $C_{i,j}=1$ if connected component $i$ contains atom
$j$ and $C_{i,j}=0$ otherwise. Then by Lemma \ref{thm:splittingN},
we have 

\begin{equation}
T=\mathrm{\textrm{diag}^{-1}}\left(C^{T}\mathbbm{1}\right)\sum_{i=1}^{c}T(i),\label{eq:Asplit-1}
\end{equation}
where $T(i)\in\left\{ -1,0,1\right\} ^{p\times q}$ is an incidence
matrix for the $i^{th}$ connected component of $\mathcal{G}(\mathcal{X},\mathcal{\mathcal{E}},\mathcal{H})$,
given by
\begin{equation}
T(i)\coloneqq\textrm{diag}(C_{i,:})T.\label{def:A(k)-1}
\end{equation}
Next, we will show how connected components that are identical in
particular ways may be identified.

\subsection{\label{subsec:Label-preserving-isomorphism}Isomorphic connected
components of an atom transition graph}

We define a pair of connected components in an atom transition graph
to be isomorphic, under a label-preserving isomorphism, if their incidence
matrices are permutationally equivalent and the molecular species
label of each atom is preserved. Henceforth, for brevity, we denote
a label-preserving isomorphism simply as an isomorphism. A maximal
subgraph isomorphism class of an atom transition graph is a maximal
set of pairwise isomorphic connected components of that graph. Each
conserved moiety corresponds to one maximal subgraph isomorphism class
of an atom transition graph. Each atom in a conserved moiety corresponds
to a distinct connected component in a maximal subgraph isomorphism
class of an atom transition graph. The number of atoms in a conserved
moiety is equal to the number of connected components in the corresponding
maximal subgraph isomorphism class of an atom transition graph. An
instance of a conserved moiety is composed of atoms, each of which
have the same molecular label. The number of conserved moieties is
equal to the number of maximal subgraph isomorphism classes of connected
components of an atom transition graph $\mathcal{A}(\mathcal{X},\mathcal{\mathcal{E}},\mathcal{H})$.

\subsection{\label{subsec:Example-of-conserved}Example isomorphism classes of
an atom transition graph}

Figure \ref{fig:iso_lable_moities} illustrates the two conserved
moieties of the 3 reaction biochemical network introduced in Section
\ref{subsec:Example-directed-stoichiometric} 

\begin{figure}[H]
\begin{centering}
\includegraphics[scale=0.95]{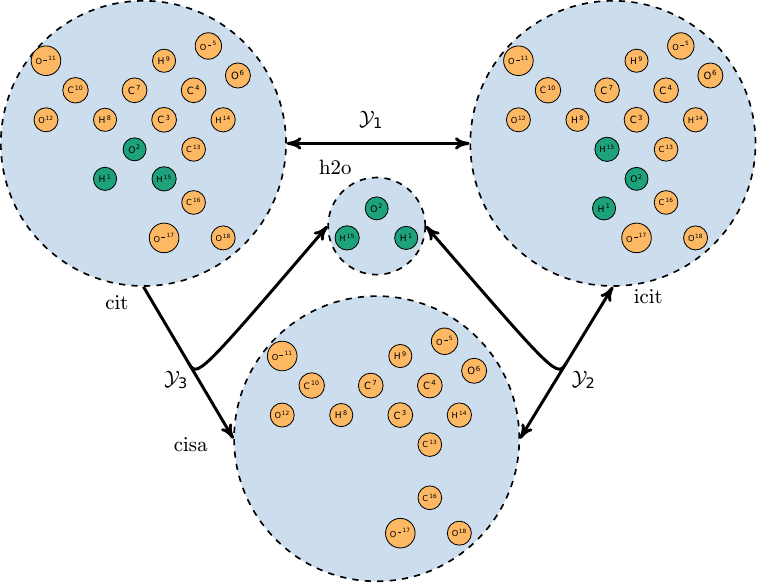}
\par\end{centering}
\caption{\label{fig:iso_lable_moities}\textbf{Maximal isomorphism classes
of an atom transition graph. }Each molecular species in the set $\{icit,\:h2o,\:cit,\:cisa\}$
is displayed as a set of atoms, without considering bonds. Connected
components corresponding to atoms 1, 2 and 15 (green) form one maximal
isomorphism class. In a metabolite, the set of atoms $\{H^{1},O^{2},H^{15}\}$
correspond to an instance of one conserved moiety. Connected components
corresponding to atoms 4, 5, 6, 8, 9, 10, 11, 12, 13, 14, 3, 16, 17,
and 18 (yellow) form another maximal isomorphism class and therefore
another conserved moiety.}
\end{figure}

\subsection{\label{subsec:Conserved-moiety-structure}Conserved moiety topology}

Thus far we have identified the atoms but not yet the bonds within
a conserved moiety. To completely identify a chemical (sub)topology
of a molecular species that remains invariant with respect to the
chemical transformations in a given biochemical network we also need
to identify the bonds within a conserved moiety. Consider a conserved
molecular transition graph $\bar{\mathcal{L}}(\mathcal{X},\bar{\mathcal{B}},\mathcal{E},\mathcal{H})$,
where each vertex is an atom and each edge is either an atom transition
or a conserved bond\textit{. Contract} each subgraph of the conserved
molecular transition graph $\bar{\mathcal{L}}(\mathcal{X},\bar{\mathcal{B}},\mathcal{E},\mathcal{H})$
that is connected by a set of atom transitions, into a single vertex
to generate a condensed conserved molecular graph $\bar{\mathcal{L}}(\mathcal{X},\bar{\mathcal{B}},\mathcal{H})$
(cf graph condensation in Section \ref{subsec:Graph-condensation}).
Each vertex of $\bar{\mathcal{L}}(\mathcal{X},\bar{\mathcal{B}},\mathcal{H})$
results from contraction of a connected component of an atom transition
graph $\mathcal{T}(\mathcal{X},\mathcal{\mathcal{E}},\mathcal{H})$
and now represents an atom in a conserved moiety. Each edge of $\bar{\mathcal{L}}(\mathcal{X},\bar{\mathcal{B}},\mathcal{H})$
results from contraction of a conserved bond transition graph and
now represents a bond in a conserved moiety. 

Let $\bar{\underline{A}}$ denote the incidence matrix of the condensed
conserved graph $\bar{\mathcal{L}}(\mathcal{X},\bar{\mathcal{B}},\mathcal{H})$,
which is obtained by the following condensation
\[
\bar{\underline{A}}\coloneqq\textrm{diag}\left(\left[\begin{array}{c}
d\left(\bar{A}\right)\\
d\left(T\right)
\end{array}\right]\right)^{-1}\cdot\left[\begin{array}{cc}
C\left(\bar{A}\right) & 0\\
0 & C\left(T\right)
\end{array}\right]\cdot\left[\begin{array}{cc}
\bar{A} & T\end{array}\right],
\]
where the entries of $d\left(\bar{A}\right)$ equal the number of
conserved bonds in the corresponding conserved bond transition graph,
the entries with $d\left(T\right)$ equal to the number of atom transitions
in the corresponding connected component, $C\left(\bar{A}\right)\in\{0,1\}$
is a matrix that maps each bond in a conserved bond transition graph
to a bond in the conserved molecular transition graph with incidence
matrix $\left[\begin{array}{cc}
\bar{A} & T\end{array}\right]$, and $C\left(T\right)\in\{0,1\}$ is a matrix that maps each connected
component to an atom transition of the conserved molecular transition
graph. Each connected component of the condensed conserved graph corresponds
to a distinct conserved moiety and the topology of each component
identifies the molecular topology of a conserved moiety.

All conserved moiety instances of the same conserved moiety are structurally
identical up to a permutation of their vertices (atoms) and edges
(bonds). Therefore a conserved moiety is a maximal isomorphism class
of conserved moiety instances. Formally, a conserved moiety is a maximal
isomorphism class 

\[
\mathcal{Q}_{k}\coloneqq\{\bigcup_{i\in\mathcal{I}(k)}\mathcal{Q}_{k}(\mathcal{X},\bar{\mathcal{B}},\mathcal{V}_{i})\},\text{}\text{\ensuremath{k\in\{1,...,\left|\mathcal{I}\right|\}}}.
\]
where $\mathcal{Q}_{k}(\mathcal{X},\mathcal{B},\mathcal{V}_{i})$
a conserved moiety instance in molecular species $\mathcal{V}_{i}$,
with $\left|\mathcal{X}\right|$ vertices and $\left|\bar{\mathcal{B}}\right|$
conserved bonds. Each conserved moiety instance $\mathcal{Q}_{k}(\mathcal{X},\bar{\mathcal{B}},\mathcal{V}_{i})$
is represented by an incidence matrix $Q_{k,i}\in\{-1,0,1\}^{\left|\mathcal{X}\right|\times\left|\bar{\mathcal{B}}\right|}$
that defines its molecular topology. In Section \ref{subsec:Label-preserving-isomorphism}
we stated that an instance of a conserved moiety is composed of atoms,
each of which have the same molecular label. It is possible that a
molecule contains more than one instance of the same conserved moiety
and when the topology of a conserved moiety contains more than one
atom, it is the bond(s) between atoms in each conserved moiety instance
that enables one to distinguish which atoms are part of which instance.

\subsection{\label{subsec:Example-of-conserved-structure}Example conserved moiety
topology}

Figure \ref{fig:Conserved-moieties-molecular} illustrates the molecular
topology of two conserved moieties.

\begin{figure}[H]
\begin{centering}
\includegraphics[scale=0.35]{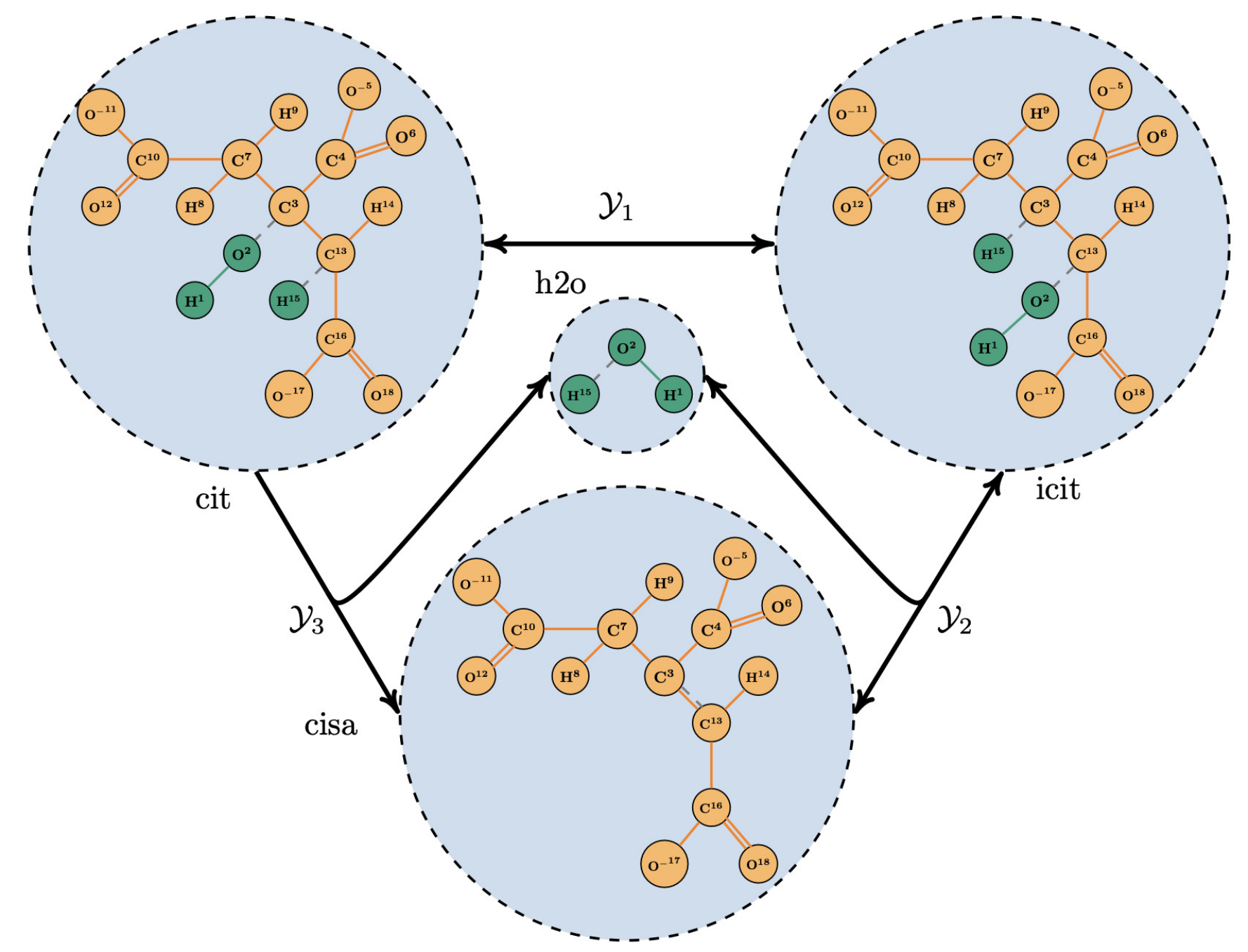}
\par\end{centering}
\caption{\textbf{\label{fig:Conserved-moieties-molecular}Conserved moiety
topology. }Each conserved moiety instance is represented by a set
of atoms and bonds. The molecular graph of the conserved moiety instance
(yellow) in the metabolite $cit$ is isomorphic to the molecular graphs
of the conserved moiety instances (yellow) in metabolites $icit$
and $cisa$. The molecular graph of the conserved moiety instance
(green) in the metabolite $cit$ is isomorphic to the molecular graphs
of the conserved moiety instances (green) in metabolites $icit$ and
$h2o$. Dashed edges represent reacting bonds, which are not part
of any conserved moiety because they are part of the reacting molecular
transition graph.}
\end{figure}

\subsection{\label{subsec:Moiety-transition-graph}Moiety transition graph}

Let $\mathcal{I}$ denote the set of maximal subgraph isomorphism
classes of an atom transition graph, and $\left|\mathcal{I}\right|$
denote the number of maximal isomorphism classes, where $k\in\{1,...,\left|\mathcal{I}\right|\}$
is an index of a maximal isomorphism class. Let $H\in\{0,1\}^{\left|\mathcal{I}\right|\times c}$
denote a mapping between $\left|\mathcal{I}\right|$ isomorphism classes
and $c$ connected components, where $H_{k,i}=1$ if isomorphism class
$k$ contains connected component $i$ and $H_{k,i}=0$ otherwise.
Let $T\in\left\{ -1,0,1\right\} ^{p\times q}$ be an incidence matrix
for an atom transition graph, then
\[
T(k,i)\coloneqq\textrm{diag}(\textrm{diag}(H_{k})\cdot C_{i})\cdot T,
\]
is the incidence matrix of the $i^{th}$ connected component of the
atom transition matrix corresponding to the $k^{th}$ isomorphism
class $\mathcal{I}(k)$, with $T(k,i)\in\{-1,0,1\}^{p\times q}$.
If the $j^{th}$ connected component of the atom transition matrix
is not part of the $k^{th}$ isomorphism class, then $T(k,j)=\{0\}^{p\times q}$.
Let component $i$ and component $j$ of the atom transition graph
belong to the same isomorphism class then there exists a label-preserving
permutation matrix between $P(i,j)\in\{0,1\}^{p\times q}$ the $i^{th}$
and $j^{th}$ connected components, such that
\[
T(k,i)=P(i,j)^{\mathrm{\mathrm{}}}T(k,j)P(i,j)^{T}
\]
which maps rows to each other that have with identical metabolite
labels. Since each connected component within an isomorphism class
has permutationally equivalent topology, we can arbitrarily choose
one incidence matrix of a connected component of the atom transition
matrix to represent the topology of each isomorphism class. This canonical
incidence matrix for the $k^{th}$ isomorphism class is denoted $T(k,\circ)$. 

Next, we show how this incidence matrix provides the topology for
the set of feasible transitions of a conserved moiety instance between
pairs of substrate and product metabolites. A moiety transition graph
$\mathcal{M}(\mathcal{X},\mathcal{\mathcal{E}},\mathcal{H},\mathcal{A})$
is a directed graph where each vertex is a conserved moiety instance
and each edge is a moiety transition between a conserved moiety instance
in a substrate molecular species and another conserved moiety instance,
of the same conserved moiety, in a product molecular species. A moiety
transition graph consists of $\left|\mathcal{I}\right|$ connected
components, each corresponding to one conserved moiety and each corresponding
to one maximal isomorphism class of an atom transition graph. In a
\emph{moiety transition graph incidence matrix} of the $k^{th}$ connected
component is

\[
M(k)\coloneqq T(k,\circ)=\left(\frac{1}{\left|\mathcal{I}(k)\right|}\right)\left(\sum_{j=1}^{\left|\mathcal{I}(k)\right|}P(i,j)^ {}T(k,j)P(i,j)^{T}\right),
\]
where the $k^{th}$ maximal isomorphism class $\mathcal{I}(k)$ of
atom transition graph $\mathcal{T}(\mathcal{X},\mathcal{\mathcal{E}},\mathcal{H})$
consists of $\left|\mathcal{I}(k)\right|$ connected components. That
is, $M(k)$ is identical to the canonical incidence matrix for the
$k^{th}$ maximal isomorphism class $\mathcal{I}(k)$ and permutationally
equivalent to each connected component in that class, where the molecular
species label of each atom is preserved. Since a moiety transition
graph $\mathcal{G}(\mathcal{X},\mathcal{\mathcal{E}},\mathcal{H})$
consists of $\left|\mathcal{I}\right|$ connected components, the
incidence matrix of a moiety transition graph $\mathcal{M}(\mathcal{X},\mathcal{\mathcal{E}},\mathcal{H},\mathcal{T})$
is 
\begin{equation}
M\coloneqq\sum_{k=1}^{\left|\mathcal{I}\right|}M(k)\label{eq:Msplit}
\end{equation}
where $M(k)$ is the incidence matrix of the $k^{th}$ connected component,
and $\left|\mathcal{I}\right|$ is the number of maximal isomorphism
classes of the corresponding atom transition graph.

When an atom or atom transition does not participate in an isomorphic
component, then the corresponding row or column of $P$ is all zeros,
respectively. It follows that $P(i,j)=0^{p\times q}$ if the $i^{th}$
and $j^{th}$ connected components are not isomorphic. As defined
above, the incidence matrix of a moiety transition graph has the same
dimensions as the incidence matrix of an atom transition graph. However,
because each conserved moiety is typically formed from more than one
connected component, one can remove its zero rows and columns and
define an incidence matrix $M\in\left\{ -1,0,1\right\} ^{u\times v},$
between a set of $u\coloneqq\left|\mathcal{X}\right|\ll p$ vertices,
each of which is a conserved moiety instance in a particular molecular
species, and $v\coloneqq\left|\mathcal{E}\right|\ll q$ edges, each
of which is a moiety transition. 

\subsection{\label{subsec:Conserved-moiety-splitting}Moiety graph decomposition
of a stoichiometric matrix}

Section \ref{subsec:Moiety-transition-graph} established a relationship
between a moiety transition graph $\mathcal{M}(\mathcal{X},\mathcal{\mathcal{E}},\mathcal{H},\mathcal{T})$
and an atom transition graph\textit{ }$\mathcal{T}(\mathcal{X},\mathcal{E},\mathcal{H})$.
This section establishes a relationship between a conserved moiety
transition graph and a stoichiometric hypergraph $\mathcal{H}(\mathcal{V},\mathcal{Y}(\mathcal{S},\mathcal{P}))$.
To this end we define two mapping matrices as follows. Let $V\in\left\{ 0,1\right\} ^{m\times u}$
denote a matrix that maps each metabolite to each conserved moiety
instance, that is $V_{i,j}=1$ if metabolite $i$ contains conserved
moiety instance $j$, and $V_{i,j}=0$ otherwise. Each column of $V$
contains a single $1$ since each conserved moiety instance is labelled
with a molecular label and is therefore specific to a particular metabolite.
Let $E\in\left\{ -1,0,1\right\} ^{v\times n}$ denote a matrix that
maps each conserved moiety transition to each reaction, that is $E_{i,j}=1$
if moiety transition $i$ occurs with the same orientation in reaction
$j$ , $E_{i,j}=-1$ if moiety transition $i$ occurs with the opposite
orientation in reaction $j$ and $E_{h,j}=0$ otherwise. 

The internal stoichiometric matrix $N$ can be expressed in terms
of $M$, $V$, and $E$ by
\begin{eqnarray}
N & = & \left(VV^{\mathrm{T}}\right)^{-1}VME.\label{eq:Decomposition}
\end{eqnarray}
Each column of $V$ contains a single $1$ so the matrix $\left(VV^{\mathrm{T}}\right)\in\mathbb{N}_{0}^{m\times m}$
is a diagonal matrix with the total number of moiety instances in
each metabolite along the diagonal. It is important to be clear that
the total number of moiety instances may consist of moiety instances
of more than one moiety. The right hand side of Eq. \ref{eq:VME}
is therefore the internal stoichiometric matrix with each row scaled
by the total number of instances of all moieties in the corresponding
metabolite. Every metabolite contains at least one moiety so $\left(VV^{\mathrm{T}}\right)$
is invertible. The decomposition in Eq. \ref{eq:Decomposition} can
more easily be interpreted by rearranging terms to obtain
\begin{equation}
\left(VV^{\mathrm{T}}\right)N=VME.\label{eq:VME}
\end{equation}
Inserting \ref{eq:Msplit} into \ref{eq:VME}, one obtains the following
decomposition of a stoichiometric matrix
\begin{eqnarray}
N & = & \left(VV^{\mathrm{T}}\right)^{-1}V\left(\sum_{k=1}^{\left|\mathcal{I}\right|}M(k)\right)E\nonumber \\
 & = & \left(VV^{\mathrm{T}}\right)^{-1}\sum_{k=1}^{\left|\mathcal{I}\right|}N(k)\label{eq:VVtNk}
\end{eqnarray}
where $N(k)$ is the $k^{th}$ moiety transition matrix, given by
\[
N(k)\coloneqq VM(k)E.
\]
Section \ref{sec:AppConserved-moiety-splitting} establishes a correspondence
between this conserved moiety decomposition of a stoichiometric matrix
and a conserved moiety splitting of a stoichiometric matrix established
previously \cite{ghaderi_structural_2020}.

\subsubsection{Example of conserved moiety splitting of a stoichiometric matrix}

Let $V\in\{0,1\}^{4\times6}$ represent the matrix that maps each
metabolite\textit{ of the network represented in Figure }\ref{subsec:Example-directed-stoichiometric}\textit{
to each moiety instance. Each row corresponds to a metabolite, and
each column corresponds to a conserved moiety instance.}

\begin{table}[H]
\centering{}%
\begin{tabular}{ll|llllllll|l}
 & \multicolumn{1}{l}{} &  &  &  &  &  &  &  & \multicolumn{1}{l}{} & \tabularnewline
 & \multicolumn{1}{l}{} &  & $L_{1}(h2o)$ & $L_{1}(cit)$ & $L_{1}(icit)$ & $L_{2}(cit)$ & $L_{2}(icit)$ & $L_{2}(cisa)$ & \multicolumn{1}{l}{} & \tabularnewline
\cline{3-3}\cline{10-10}
\multirow{4}{*}{$V=$} & $h2o$ &  & 1 & 0 & 0 & 0 & 0 & 0 &  & \multirow{4}{*}{}\tabularnewline
 & $cit$ &  & 0 & 1 & 0 & 1 & 0 & 0 &  & \tabularnewline
 & $icit$ &  & 0 & 0 & 1 & 0 & 1 & 0 &  & \tabularnewline
 & $cisa$ &  & 0 & 0 & 0 & 0 & 0 & 1 &  & \tabularnewline
\cline{3-3}\cline{10-10}
 & \multicolumn{1}{l}{} &  &  &  &  &  &  &  & \multicolumn{1}{l}{} & \tabularnewline
\end{tabular}
\end{table}

Similarly, let $E\in\{-1,0,1\}^{6\times3}$ denote the matrix that
maps each moiety transition to each reaction in\textit{ }Figure\textit{
}\ref{subsec:Example-directed-stoichiometric}. Here, each row represents
a conserved moiety transition, and each column represents a reaction.

\begin{table}[H]
\centering{}%
\begin{tabular}{ll|lllll|l}
 & \multicolumn{1}{l}{} &  & $\mathcal{Y}_{1}$ & $\mathcal{Y}_{2}$ & $\mathcal{Y}_{3}$ & \multicolumn{1}{l}{} & \tabularnewline
\cline{3-3}\cline{7-7}
 & $M_{1}$ &  & 1 & 0 & 0 &  & \tabularnewline
\multirow{4}{*}{$E=$} & $M_{2}$ &  & 1 & 0 & 0 &  & \multirow{2}{*}{}\tabularnewline
 & $M_{3}$ &  & 0 & 1 & 0 &  & \tabularnewline
 & $M_{4}$ &  & 0 & 1 & 0 &  & \tabularnewline
 & $M_{5}$ &  & 0 & 0 & 1 &  & \tabularnewline
 & $M_{6}$ &  & 0 & 0 & 1 &  & \tabularnewline
\cline{3-3}\cline{7-7}
\end{tabular}
\end{table}

Furthermore, let $M\in\{-1,0,1\}^{6\times6}$ be the incidence matrix
of the moiety defined in Figure \ref{fig:Conserved-moieties-molecular}.
In this matrix, each row represents a conserved moiety instance and
each column represents a conserved moiety transition.

\begin{table}[H]
\centering{}%
\begin{tabular}{ll|llllllll|l}
 & \multicolumn{1}{l}{} &  & $M_{1}$ & $M_{2}$ & $M_{3}$ & $M_{4}$ & $M_{5}$ & $M_{6}$ & \multicolumn{1}{l}{} & \tabularnewline
\cline{3-3}\cline{10-10}
 & $L_{1}(h2o)$ &  & 0 & 0 & 1 & 0 & 1 & 0 &  & \tabularnewline
\multirow{4}{*}{$M=$} & $L_{1}(cit)$ &  & -1 & 0 & 0 & 0 & -1 & 0 &  & \multirow{2}{*}{}\tabularnewline
 & $L_{1}(icit)$ &  & 1 & 0 & -1 & 0 & 0 & 0 &  & \tabularnewline
 & $L_{2}(cit)$ &  & 0 & -1 & 0 & 0 & 0 & -1 &  & \tabularnewline
 & $L_{2}(icit)$ &  & 0 & 1 & 0 & -1 & 0 & 0 &  & \tabularnewline
 & $L_{2}(cisa)$ &  & 0 & 0 & 0 & 1 & 0 & 1 &  & \tabularnewline
\cline{3-3}\cline{10-10}
\end{tabular}
\end{table}

Finally, the sample matrix calculation shows that $\left(VV^{T}\right)N=VME.$

\section{Reacting moieties}

In this section, we shall define a \textit{reacting moiety} as a particular
set of reacting bonds. Although each chemical reaction involves a
set of reacting bonds, some reactions share reacting bonds that are
isomorphic up to labelling of their associated reactions, so we aim
to identify a minimal set of reactions that cover all reacting bonds
and then define a reacting moiety as a set of reacting bonds corresponding
to a reaction in that minimal set. Therefore, first we condense similar
parts of the reacting molecular transition graph, then we formulate
a minimal set cover problem to identify a minimal number of reactions
and hence a minimal number of reacting moieties.

\subsection{Condensation of a reacting molecular transition graph}

Consider a reacting molecular transition graph $\hat{\mathcal{\mathcal{L}}}(\mathcal{X},\hat{\mathcal{B}},\mathcal{E},\mathcal{H})$,
where each vertex is an (ambivorous) atom and each edge is either
an atom transition or a reacting bond\textit{. Contract} each component
of the reacting molecular transition graph $\hat{\mathcal{\mathcal{L}}}(\mathcal{X},\hat{\mathcal{B}},\mathcal{E},\mathcal{H})$
that is strongly connected by a set of atom transitions, into a single
vertex to generate a reacting moiety graph $\hat{\mathcal{\mathcal{L}}}(\underline{\mathcal{X}},\hat{\mathcal{B}},\mathcal{E},\mathcal{H})\coloneqq\underline{\hat{\mathcal{L}}}$
(cf graph condensation in Section \ref{subsec:Graph-condensation}).
Each vertex of $\underline{\hat{\mathcal{L}}}$ represents a strongly
connected component of the atom transition graph $\mathcal{T}(\mathcal{X},\mathcal{\mathcal{E}},\mathcal{H})$.
Each edge of $\underline{\hat{\mathcal{L}}}$ corresponds to a reacting
bond between a pair of vertices, each representing a contracted connected
component of an atom transition graph. Each edge is labelled with
the set of reactions corresponding to the reacting bond (i.e. where
that bond is broken, formed, or changes order). This condensation
may be represented by
\[
\underline{\hat{B}}\coloneqq\textrm{diag}(d)^{-1}\cdot C\cdot\hat{B},
\]
where $d\in\mathbb{N}^{p}$ with $d_{i}$ equal to the number of atoms
in the $i^{th}$ connected component, $C\in\{0,1\}^{p\times c}$ is
a matrix that maps each connected component to an atom of the molecular
transition graph with incidence matrix $\hat{B}$, and $\underline{\hat{B}}$
is the incidence matrix of the reacting moiety graph $\underline{\hat{\mathcal{L}}}$,
where each edge is a reacting bond, represented without considering
molecular specificity. It is important to note that in the case of
a single reaction, the reacting moiety graph $\underline{\hat{\mathcal{L}}}$
is equivalent to the graph derived from the reaction matrix defined
in Section \ref{subsec:reactionMatrix}. 

\subsection{Example of a reacting moiety graph}

Figure \ref{fig:Condensed-reacting-graph} illustrates the reacting
moiety graph of the reaction shown in Figure \ref{fig:Substrate and product matrices}.
It is important to note that this graph is equivalent to the graph
derived from the reaction matrix $D$, without considering the sign
of its entries.

\begin{figure}
\selectlanguage{american}%
\hspace{6cm}\foreignlanguage{british}{\includegraphics[scale=0.3]{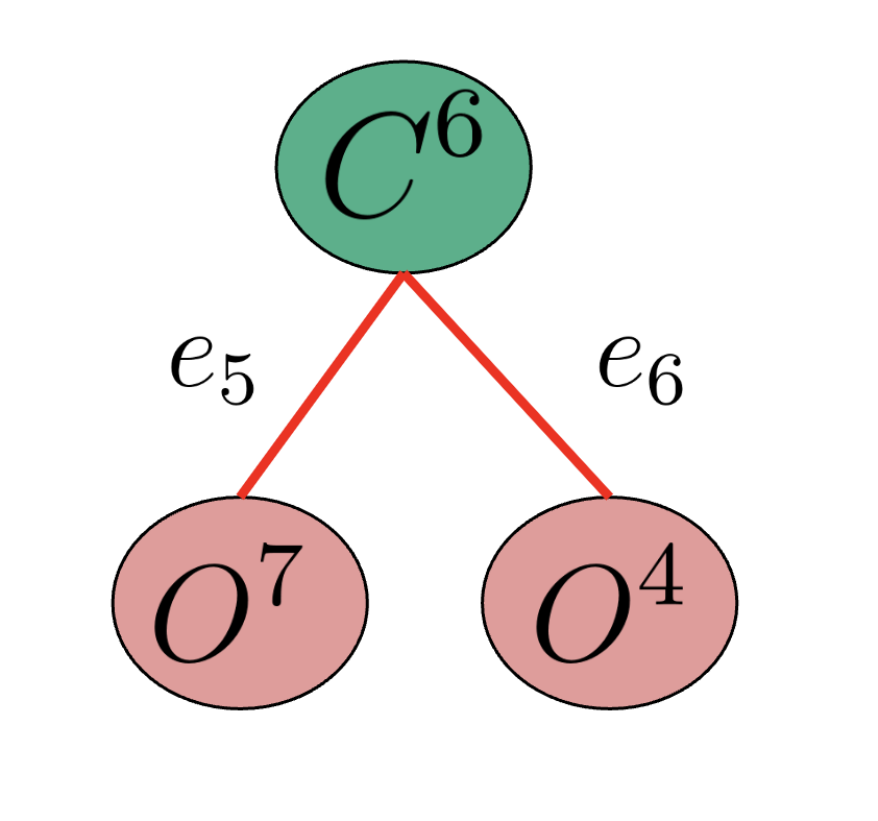}\caption{\textbf{\label{fig:Condensed-reacting-graph}Reacting moiety graph.
}Each atom transition component is condensed into a single node representing
an atom without molecular specification, while each edge represents
a reacting bond in the reaction (see Figure \ref{fig:Substrate and product matrices}).}
}\selectlanguage{british}%
\end{figure}

\subsection{Minimal set cover of a reacting moiety graph}

The minimal number of reactions to cover all associated reacting bonds
in the reacting moiety graph can be obtained from a solution to a
minimal set cover problem. Consider a reacting moiety graph $\underline{\hat{\mathcal{L}}}$,
where each vertex represents a contracted ambivorous atom and each
edge represents a reacting bond. For practical implementation, this
problem can be formulated in matrix form by defining an incidence
matrix $A\in\{0,1\}^{d\times n}$, where 

\[
A_{ij}\coloneqq\begin{cases}
1 & \text{if bond \ensuremath{b_{j}}is involved in reaction\ensuremath{\mathcal{Y}_{i}}},\\
0 & \textrm{otherwise.}
\end{cases}
\]

Let $x\in\{0,1\}^{n}$be a binary decision vector, where 

\[
x_{i}\coloneqq\begin{cases}
1 & \mathcal{Y}_{i}\text{is in the cover},\\
0 & \textrm{otherwise.}
\end{cases}
\]

Also, let $b\in\mathbb{N}^{d\times1}$ be a vector of ones, since
every bond must be covered at least once. The minimal set cover problem
can then be mathematically formulated as

\[
\begin{array}{c}
\underset{x\in\{0,1\}^{n}}{min}\mathds{1}^{T}x,\\
s.t.\;Ax\geqslant b,
\end{array}
\]
where $\mathds{1}^{T}$ is the transpose of the all ones vector. The
objective function minimises the number of reactions needed to cover
all reacting bonds, ensuring each reacting bond is included in at
least one selected reaction. Solving this integer linear programming
(ILP) problem identifies a minimal reaction set, where each reaction
in that minimal set identifies a reacting moiety by the reacting bonds
it is associated with.

\subsection{\label{subsec:Example-of-reacting}Example of reacting moieties}

Figure \ref{fig:Reacting-transitions-graph.} illustrates the reacting
moiety graph of the biochemical network introduced in Section \ref{subsec:Example-directed-stoichiometric}.
Each vertex represents a contracted ambivorous atom and each edge
represents a reacting bond, associated with one or more reactions
as illustrated in Figure \ref{fig:Reacting-moieties-molecular}. Each
reaction within this network is represented as a set of reacting bonds
that either break or form during a reaction. The minimal set cover
of this graph identifies the minimal subset of reactions that cover
all reacting bonds. 

\selectlanguage{american}%
\begin{figure}[H]
\hspace{3cm}\includegraphics[scale=0.3]{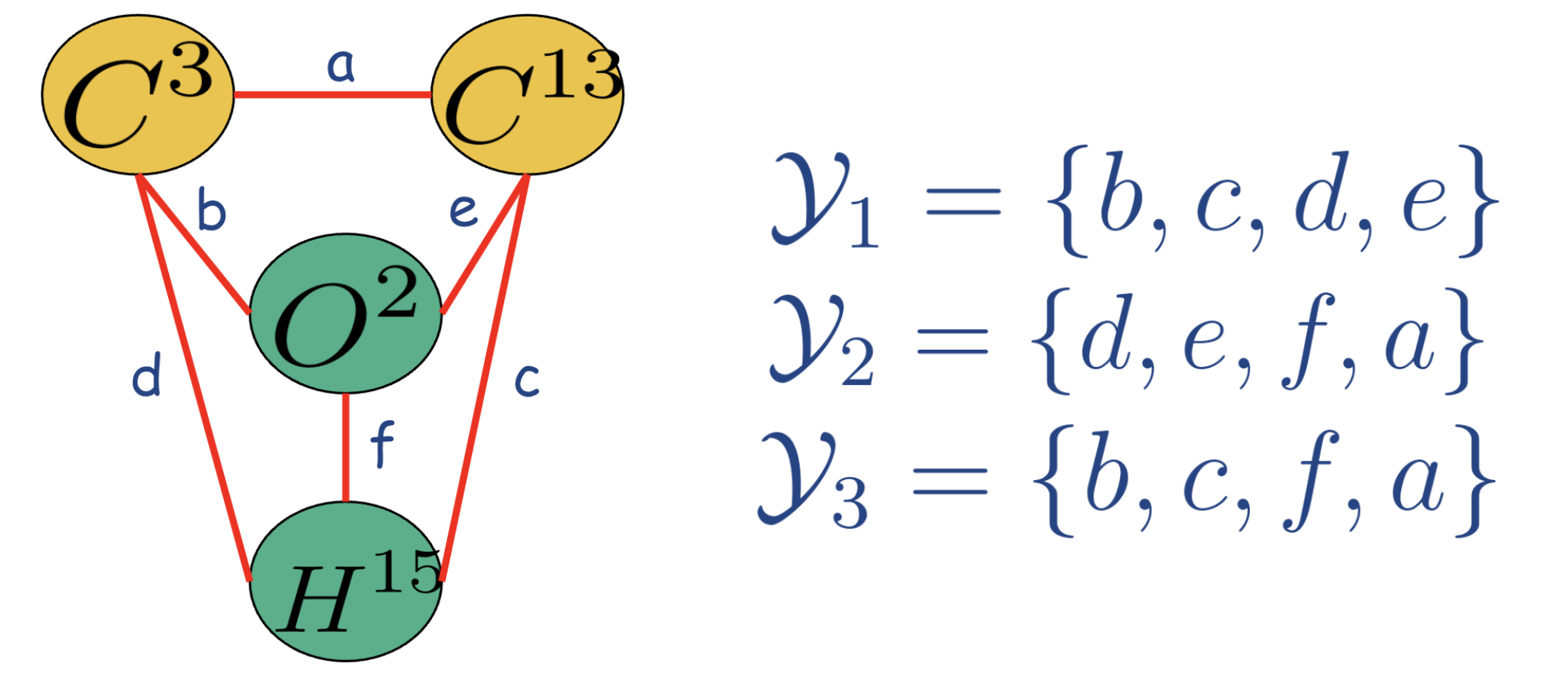}

\selectlanguage{british}%
\caption{\textbf{\label{fig:Reacting-transitions-graph.}Reacting moiety graph.
}Each vertex represents a contracted ambivorous atom and each edge
corresponds to a reacting bond. Each reaction is defined by the reacting
bonds involved in the transformation. For instance, in reaction $\mathcal{Y}_{1}$
, two bonds (b:$C^{3}-O^{2}$ and c: $C^{13}-H^{15}$) are broken,
while two new bonds (\ensuremath{d}: $C^{3}-H^{15}$ and \ensuremath{e}:
$C^{13}-O^{2}$) are formed. In reaction $\mathcal{Y}_{2}$ , two
bonds (d: $C^{3}-H^{15}$ and e: $C^{13}-O^{2}$ are broken, and two
new bonds (f: $O^{2}-H^{15}$ and a: $C^{3}-C^{13}$) are formed.
Lastly, in reaction $\mathcal{Y}_{3}$ , two bonds (b: $C^{3}-O^{2}$
and c: $C^{13}-H^{15}$) are broken, while two new bonds (f: $O^{2}-H^{15}$
and a : $C^{3}-C^{13}$) are formed. The set of reacting bonds for
reaction $\mathcal{Y}_{1}=\{b,c,d,e\}$ and reaction $\mathcal{Y}_{3}=\{b,c,f,a\}$
cover all the reacting bonds in the network. The reacting bonds corresponding
to the reaction $\mathcal{Y}_{1}$ between the pairs of atoms $(C^{3}-O^{2})$,
$(C^{13}-H^{15})$, ($C^{3}-H^{15})$, and $(C^{13}-O^{2})$ form
one reacting moiety illustrated in Figure \ref{fig:Reacting-moieties-of}
(a), and the reacting bonds corresponding to the reaction $\mathcal{Y}_{3}$
between the pairs of atoms $(C^{3}-O^{2})$, $(C^{13}-H^{15})$,$(O^{2}-H^{15})$,
and $(C^{3}-C^{13})$ form the second reacting moiety represented
in Figure \ref{fig:Reacting-moieties-of} (b). Note that (a: $C^{3}-C^{13}$
and f: $O^{2}-H^{15}$) are between atoms of a conserved moiety, while
the other reacting bonds are between instances of a pair of distinct
conserved moieties.}
\selectlanguage{american}%
\end{figure}

\selectlanguage{british}%
\begin{figure}
\subfloat[]{\includegraphics[scale=0.7]{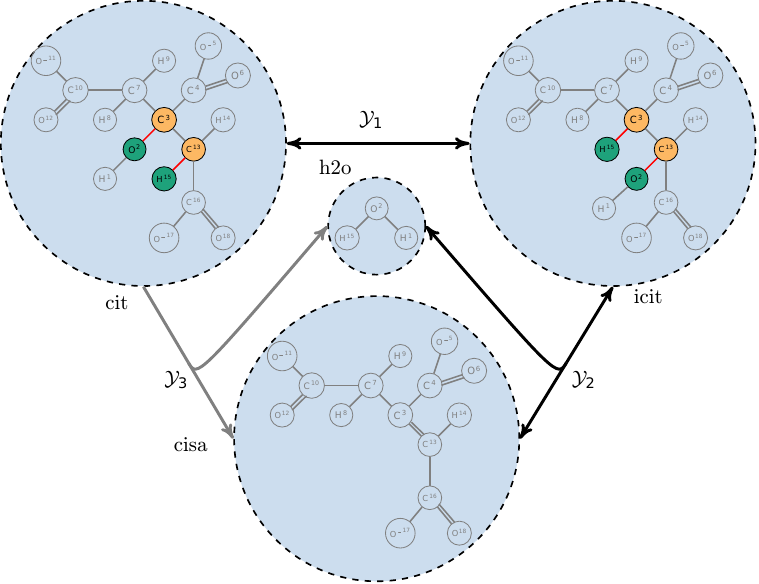}

}\subfloat[]{\includegraphics[scale=0.7]{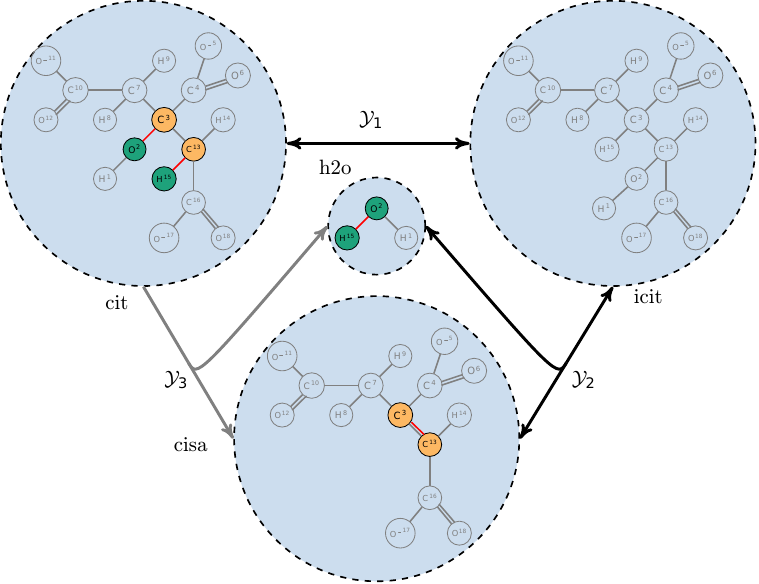}}\caption{\textbf{\label{fig:Reacting-moieties-of}Reacting moieties of a molecular
graph. }(a) Reacting bonds corresponding to reaction $\mathcal{Y}_{1}$
form the first reacting moiety. (b) Reacting bonds corresponding to
reaction $\mathcal{Y}_{3}$ make up the second reacting moiety. Atoms
are highlighted to emphasise particular bonds, but an atom is not
a component of a reacting moiety.}
\end{figure}

\section{Discussion}

\paragraph{Characterisation of conserved moieties }

Previously, we developed methods to identify the atoms in a conserved
moiety \cite{haraldsdottir_Identification_2016} and the set of conserved
moieties for a given network \cite{ghaderi_structural_2020} but
the structure of each conserved moiety was not specified. Herein we
identify the topology of each conserved moiety, in terms of conserved
bonds, that are invariant with respect to all of the chemical transformations
in a network. Previously, we demonstrated how a stoichiometric matrix
could be split into the sum of a set of moiety transition matrices
\cite{ghaderi_structural_2020}. However, there was no guarantee
that each moiety transition matrix corresponded to an incidence matrix
of a graph. Herein, we introduce a moiety transition graph, whose
incidence matrix is a graph and each connected component of a moiety
transition graph corresponds to a distinct conserved moiety. 

\paragraph{Characterisation of reacting moieties}

We presented the first linear algebraic and graph theoretical definition
of a reacting moiety, in terms of reacting bonds, that are either
broken or formed by at least one reaction in a network. This contrasts
with established approaches is that define reaction centres, reaction
sites, or the like, in heuristic terms that do not admit an unambiguous
mathematical interpretation. We introduced the novel concept of a
reacting moiety graph, where each vertex is an atom and each edge
corresponds to a bond that is either broken or formed in a network.
While we use it to identify a minimal set of reacting moieties, it
is envisaged to lead to novel theoretical applications, e.g., estimation
of thermodynamic properties of biochemical networks. 

\paragraph{Hypergraphs versus graphs}

The use of an atom and molecular transition graphs instead of a stoichiometric
hypergraph alone is motivated by the theoretical and computational
benefits offered by working with graphs. Graphs theory provides a
well established and comprehensive theoretical framework with numerous
algorithms optimised to efficiently solve a wide variety of problems
involving graphs, e.g., graph isomporphism, minimal set cover. In
contrast, in general, a hypergraph may have arbitrarily complex topology
so with less structure to exploit, there are comparatively far fewer
theoretical results and algorithms available for solving problems
involving hypergraphs.

From a biochemical perspective, it is natural to consider a graph
of conserved moiety transitions, as by definition, a conserved moiety
is an invariant chemical (sub)structure. We also demonstrate that
a graph is the appropriate conceptual structure to represent a reacting
moiety graph as it it built from edges representing chemical bonds.
Moreover, Section \ref{subsec:Conserved-moiety-splitting} demonstrates
that while a metabolic network is a hypergraph, its hypergraph incidence
matrix can be decomposed into a set of graphs, which is not the case
for a hypergraph in general. This has profound implications for mathematical
modelling of biochemical networks, primarily because most mathematical
modelling approaches assume a stoichiometric matrix is an arbitrary
rectangular matrix, thereby failing to exploit its special structure
to generate novel theoretical results that would not hold for arbitrary
rectangular matrices.

\paragraph{Atom mapping}

Accurate identification of conserved and reacting moieties depends
on accurate atom mappings. However, predicting accurate atom mappings
for every metabolic reaction in a genome-scale model is a challenging
cheminformatic problem due to the complexity and heterogeneity of
reaction networks. Lumped reactions, each involving a series of enzyme
catalysed reactions condensed into one reaction, should ideally be
split into a series of reactions prior to atom mapping. Molecular
symmetries can give rise to multiple valid atom mappings for a reaction,
each of which should be included. Cellular conditions can also affect
atom mappings by altering reacting mechanisms, making it difficult
to algorithmically predict the appropriate mapping for a particular
condition. Accurate identification of conserved and reacting moieties
also depends on accurate biochemical network reconstructions that
faithfully represent the underlying biochemical network. 

In our approach, we base our analysis on incorporation of information
on molecular species topology (2D MOL files) which does not take into
account geometric differences, such as stereoisomers or chirality.
This limitation could be addressed by incorporating molecular geometry
and using stereochemically-aware atom mapping algorithms that account
for spatial arrangements of atoms, such as bond angles and chiral
centres, to accurately capture stereoisomeric transformations. Future
developments in atom mapping algorithms, are necessary but not sufficient
to improve the accuracy of the atom mappings. It is also necessary
that novel atom mapping algorithms are implemented and disseminated
as accessible, interoperable and reusable software. To accurately
describe stereoisomeric transformations involving single atoms, where
there is no net cleavage or formation of bonds, the method presented
herein would need to be extended to incorporate molecular geometry.

\paragraph{Future work}

Taken together, characterisation of conserved and reacting moieties,
both in terms of their atom-bond topology and their relationship to
stoichiometric hypergraph topology provides a strong theoretical foundation,
grounded in (linear) algebraic graph theory, for novel developments
in the foundsations and applications of biochemical network analysis.
Fundamentally, it will be important to characterise, for a given biochemical
network, how the number of conserved and reacting moieties relates
to the dimensions of a stoichiometric matrix and its four fundamental
subspaces \cite{clarke_stoichiometric_1988}. 

In terms of applications, expressing a stoichiometric matrix in terms
of a set of conserved moiety graphs has already lead to the development
of conserved moiety fluxomics, a novel, efficient, mathematically
transparent, and computationally efficient method to infer metabolic
reaction flux at genome-scale\cite{fleming_conserved_2024}. Other
potential applications involve representation of reaction mechanisms
as constrained combinations of conserved and reacting moieties. For
example, certain biochemical networks result in combinatorial explosion
in the dimensions of a stoichiometric matrix. In such scenarios, a
more compact representation in terms of combinations of conserved
and reacting moieties is envisaged, since, in numerical experiments
with genome-scale metabolic networks, we observe that the number of
conserved moieties, $k$, is substantially less than the number of
molecular species, that is $k\ll m<n$ \cite{haraldsdottir_Identification_2016}.
Ideally, a reformulation in terms of conserved and reacting moieties
should be equivalent to that of a stoichiometric representation, which
will require constraints on the feasible set of moiety combinations,
e.g. combinations must be non-negative, integral and correspond to
chemically and biochemically feasible molecular topologies and reaction
mechanisms.

Despite the challenges with acquisition of sufficiently accurate input
data, particularly in large reaction networks, we emphasise the indispensability
of mathematical tools for identifying conserved and reacting moieties
in advancing our knowledge of reaction mechanisms and the behaviour
of biochemical networks Characterisation of biochemical reactions
in terms of conserved and reacting moieties opens a novel window to
further analysis of biochemical networks and bridges the gap between
graph theory, linear algebra, and biological interpretation, opening
new horizons in the study of chemical reaction networks.

More generally, the integration of established mathematical theories
and algorithms into biological systems is essential for understanding
complex biological processes. To ensure a meaningful interpretation
of the results, mathematical models must hold biological significance.
By giving biological meaning to these models, they become powerful
tools for predicting the behaviour of biological systems, which can
then be validated through real-world experiments. This feedback loop
between model predictions and experimental validation deepens our
understanding of system dynamics and enhances decision-making in various
biological applications. 

\section{Conclusion}

A conserved moiety is a chemical substructure that remains invariant
with respect to all of the chemical transformations in a chemical
reaction network. A reacting moiety is a set of bonds that are either
broken or formed in a chemical reaction network. We developed a novel
method to identify and characterise the topology of conserved and
reacting moieties in algebraic graph theoretical terms. This approach
enabled a correspondence to be established between each conserved
moiety as a member of a minimal set of distinct invariant chemical
substructures and each reacting moiety is a member of a minimal set
of distinct variant chemical substructures. Representation of a chemical
reaction network in terms of conserved and reacting moieties is a
fundamental result in the analysis of such networks. This approach
has already lead to new applications, e.g., inference of metabolic
flux by modelling the transitions of isotopically labelled conserved moieties, and is envisaged to stimulate the development of novel applications
of chemical reaction network models firmly grounded in mathematics.

\section*{Acknowledgment}

This paper is supported by European Union\textquoteright s Horizon
2020 research and innovation program under the Marie Sk\l odowska-Curie
grant agreement PoLiMeR, No 812616, the European Union's Horizon Europe
research and innovation program project Recon4IMD (grant number 101080997)
and the European Research Council (ERC) under the European Union\textquoteright s
Horizon 2020 research and innovation programme (grant number 757922
and 101125633) to Ines Thiele.

\section*{Author Contribution}

\textbf{Hadjar Rahou}, Conceptualisation, Formal analysis, Visualisation,
Writing - review \& editing; \textbf{Hulda S. Haraldsdóttir}, Conceptualisation,
Writing - review \& editing; \textbf{Filippo Martinelli}, review \&
editing; \textbf{Ines Thiele}, review \& editing; \textbf{Ronan M.T.
Fleming}, Conceptualisation, Funding acquisition, Supervision, Validation,
Writing - original draft, review \& editing.


\bibliographystyle{unsrt}


\appendix

\section{Partitioning a molecular graph}

Given a molecular graph $\mathcal{G}(\mathcal{X},\mathcal{B})$, the
corresponding molecular incidence matrix $B$, and $\left|\mathcal{I}\right|$
isomorphism classes, each atom in the molecular graph belongs to one
isomorphism class. Then, the rows of the incidence matrix $B$ are
partitioned into $\left|\mathcal{I}\right|$ partitions, where each
partition represents a set of atoms belonging to one isomorphism class.

This induces a partition of the bonds into $\left|\mathcal{I}\right|+1$
partitions. That is, each partition of $\mathcal{I}$ represents the
set of bonds in the corresponding isomorphism class. The $(\left|\mathcal{I}\right|+1)^{th}$
partition contains the reacting bonds. That is the $(\left|\mathcal{I}\right|+1)^{th}$
partition is the ``cut''; that is, they have some atoms in different
isomorphism classes. Then, without loss of generality, the molecular
graph incidence matrix $B$ may be partitioned as follows

\begin{table}[H]
\centering{}%
\begin{tabular}{ll|lllll|lll|llll|lllll|lllll}
 & \multicolumn{1}{l}{} &  & \multicolumn{10}{l}{$\text{\hspace{1.5cm}Conserved bonds\hspace{1.5cm}}$} & \multicolumn{1}{l}{} & \multicolumn{3}{l}{Reacting bonds} &  & \multicolumn{1}{l}{} &  &  &  &  & \tabularnewline
 & \multicolumn{1}{l}{} &  & $e_{1}$ & $\ldots$ & $e_{i}$ & \multicolumn{1}{l}{} & $e_{i+1}$ & $\ldots$ & \multicolumn{1}{l}{$e_{i+k'}$} & $e_{i+k'+1}$ & $\ldots$ & $e_{q-r}$ & \multicolumn{1}{l}{} & $e_{q-r+1}$ & $\ldots$ & $e_{q}$ &  & \multicolumn{1}{l}{} &  &  &  &  & \tabularnewline
\cline{3-3}\cline{19-19}
\multirow{3}{*}{} &  &  &  &  &  &  &  &  &  &  &  &  &  &  &  &  &  &  & \multicolumn{1}{l}{$a_{1}$} &  &  &  & \tabularnewline
 &  &  &  & $B_{1}$ &  &  &  & $0$ &  &  & $0$ &  &  &  &  &  &  &  & $\vdots$ &  &  &  & \tabularnewline
 &  &  &  &  &  &  &  &  &  &  &  &  &  &  &  &  &  &  & $a_{j}$ &  &  &  & \tabularnewline
\cline{4-13}
 &  &  &  &  &  &  &  &  &  &  &  &  &  &  &  &  &  &  &  &  &  &  & \tabularnewline
$B:=$ &  &  &  & $0$ &  &  &  & $\ddots$ &  &  & $0$ &  &  &  & $Z$ &  &  &  & $\vdots$ &  &  &  & \tabularnewline
 &  &  &  &  &  &  &  &  &  &  &  &  &  &  &  &  &  &  &  &  &  &  & \tabularnewline
\cline{4-13}
 &  &  &  &  &  &  &  &  &  &  &  &  &  &  &  &  &  &  & $a_{j+k}$ & \multirow{3}{*}{} &  &  & \tabularnewline
 &  &  &  & $0$ &  &  &  & $0$ &  &  & $B_{\mathcal{I}}$ &  &  &  &  &  &  &  & $\vdots$ &  &  &  & \tabularnewline
 &  &  &  &  &  &  &  &  &  &  &  &  &  &  &  &  &  &  & $a_{p}$ &  &  &  & \tabularnewline
\cline{3-3}\cline{19-19}
 & \multicolumn{1}{l}{} &  &  &  &  & \multicolumn{1}{l}{} &  &  & \multicolumn{1}{l}{} &  &  &  & \multicolumn{1}{l}{} &  &  &  &  & \multicolumn{1}{l}{} &  &  &  &  & \tabularnewline
\end{tabular}
\end{table}

The first columns of matrix $B$ correspond to conserved bonds, and
the second columns correspond to reacting bonds. Each incidence matrix
$B_{1\leq k\leq\mathcal{I}}$ represents the molecular conserved moiety
species subgraph, where each vertex is an atom and each edge is a
conserved bond, while, the matrix $Z$ represents the molecular subgraph
corresponding to the reacting bonds.

\section{\label{sec:AppConserved-moiety-splitting}Conserved moiety splitting}

\subsection{Conserved moiety matrix}

Given a stoichiometric matrix $N\in\mathbb{Z}^{m\times n}$ corresponding
to a directed stoichiometric hypergraph $\mathcal{H}(\mathcal{V},\mathcal{Y}(\mathcal{S},\mathcal{P})))$.
The conserved moiety matrix $L\in\mathbb{Z}_{+}^{r\times m}$ derived
from the corresponding atom transition graph $\mathcal{G}(\mathcal{X},\mathcal{E},\mathcal{H})$
is orthogonal to $\mathcal{R}(N)$(the column subspace), that is,
$L\cdot N=0$.

\subsubsection{Example conserved moiety matrix}

The conserved moiety matrix corresponding to Figure \ref{fig:iso_lable_moities}
is 
\begin{table}[H]
\centering{}%
\begin{tabular}{l|cccccc|l}
\multicolumn{1}{l}{} &  & $h2o$ & $cit$ & $\mathrm{icit}$ & $cisa$ & \multicolumn{1}{c}{} & \tabularnewline
\cline{2-2}\cline{7-7}
\multirow{2}{*}{$L:=$} &  & 1 & 1 & 1 & 0 &  & $L_{1}$\tabularnewline
 &  & 0 & 1 & 1 & 1 &  & $L_{2}$\tabularnewline
\cline{2-2}\cline{7-7}
\end{tabular}
\end{table}
The first and second conserved moiety vectors, $L_{1}$ and $L_{2}$
correspond to two isomorphism classes (green and yellow) in Figure
\ref{fig:iso_lable_moities}. The invariance of the number of moieties
with respect to each reaction is illustrated with 
\begin{table}[H]
\centering{}%
\begin{tabular}{lllccccll|crrrc|llll}
 &  &  &  &  &  &  &  & \multicolumn{1}{l}{} &  & $\mathcal{Y}_{1}$ & $\mathcal{Y}_{2}$ & $\mathcal{Y}_{3}$ & \multicolumn{1}{c}{} &  &  &  & \tabularnewline
\cline{10-10}\cline{14-14}
 &  &  & $h2o$ & $cit$ & $\mathrm{icit}$ & $cisa$ &  &  &  & 0 & 1 & 1 &  &  &  &  & \tabularnewline
\cline{3-3}\cline{8-8}\cline{16-16}\cline{18-18}
\multirow{2}{*}{$LN=$} & \multicolumn{1}{l|}{$L_{1}$} &  & 1 & 1 & 1 & 0 & \multicolumn{1}{l|}{} & \multirow{2}{*}{$\cdot$} &  & -1 & 0 & -1 &  & \multicolumn{1}{l|}{\multirow{2}{*}{=}} &  & 0 & \multicolumn{1}{l|}{}\tabularnewline
 & \multicolumn{1}{l|}{$L_{2}$} &  & 0 & 1 & 1 & 1 & \multicolumn{1}{l|}{} &  &  & 1 & -1 & 0 &  &  &  & 0 & \multicolumn{1}{l|}{}\tabularnewline
\cline{3-3}\cline{8-8}\cline{16-16}\cline{18-18}
 &  &  &  &  &  &  &  &  &  & 0 & 1 & 1 &  &  &  &  & \tabularnewline
\cline{10-10}\cline{14-14}
\end{tabular}
\end{table}

\subsection{Correspondence with conserved moiety splitting}

As established previously \cite{ghaderi_structural_2020}, given
an \textit{atom transition graph }$\mathcal{G}(\mathcal{X},\mathcal{\mathcal{E}},\mathcal{H})$
between a set of molecules $\mathcal{V}$, where $m\coloneqq\left|\mathcal{V}\right|$,
a conserved moiety vector $L_{k}\in\mathbb{Z}_{+}^{1\times m}$ is
a non-negative integer (row) vector, where $L_{k,i}$ is the number
of instances of the $k^{th}$ conserved moiety in molecule $\mathcal{V}_{i}$.
As there is one conserved moiety vector for each maximal graph isomorphism
class, an \textit{atom transition graph} gives rise to a set of $\left|\mathcal{I}\right|$
conserved moiety vectors, which can be concatenated to form a conserved
moiety matrix $L\in\mathbb{Z}_{+}^{\left|\mathcal{I}\right|\times m}$,
which is orthogonal to $\mathcal{R}(N)$, that is $L\cdot N=0$. Furthermore,
the following matrix splitting exists

\begin{equation}
N=\mathrm{diag^{-1}}\left(L^{T}\mathbbm{1}\right)\sum_{k=1}^{\left|\mathcal{I}\right|}N(k),\label{eq:Nsplit}
\end{equation}
where $N(k)\in\mathbb{Z}^{m\times n}$ is a moiety transition matrix,
given by
\begin{equation}
N(k)\coloneqq\textrm{diag}(L_{k})N.\label{def:N(k)}
\end{equation}
 Comparing \ref{eq:VVtNk} and \ref{eq:Nsplit}, we conclude that
\[
N=\left(VV^{T}\right)^{-1}\sum_{k=1}^{\left|\mathcal{I}\right|}N(k)=\mathrm{diag^{-1}}\left(L^{T}\mathbbm{1}\right)\sum_{k=1}^{\left|\mathcal{I}\right|}N(k)
\]
 and that 
\begin{eqnarray*}
N(k) & = & \textrm{diag}(L_{k})N=VM(k)E\\
\mathrm{diag}\left(L^{T}\mathbbm{1}\right) & = & VV^{T}
\end{eqnarray*}
establishing an equivalence between both formulations for splitting
a stoichiometric matrix.

\subsubsection{Example}
\begin{flushleft}
{\footnotesize{}%
\begin{tabular}{ll|lll|l|llll|l|lll|c}
 & \multicolumn{1}{l}{} & {\footnotesize$\mathcal{Y}_{1}$} & {\footnotesize$\mathcal{Y}_{2}$} & \multicolumn{1}{l}{{\footnotesize$\mathcal{Y}_{3}$}} & \multicolumn{1}{l}{} & {\footnotesize$h2o$} & {\footnotesize$cit$} & {\footnotesize$\mathrm{icit}$} & \multicolumn{1}{l}{{\footnotesize$cisa$}} & \multicolumn{1}{l}{} &  &  & \multicolumn{1}{l}{} & \tabularnewline
\cline{3-3}\cline{5-5}\cline{7-7}\cline{10-10}\cline{12-12}\cline{14-14}
\multirow{4}{*}{{\footnotesize$N(1)\coloneqq\textrm{diag}(L_{1})N$}} & \multirow{4}{*}{{\footnotesize$=$}} & \foreignlanguage{british}{{\footnotesize 0}} & {\footnotesize 1} & {\footnotesize 1} & \multirow{4}{*}{{\footnotesize$=$}} & {\footnotesize 1} & {\footnotesize 0} & {\footnotesize 0} & {\footnotesize 0} & \multirow{4}{*}{{\footnotesize$\cdot$}} & {\footnotesize 0} & {\footnotesize 1} & {\footnotesize 1} & {\footnotesize$h2o$}\tabularnewline
 &  & {\footnotesize -1} & {\footnotesize -1} & {\footnotesize 0} &  & {\footnotesize 0} & {\footnotesize 1} & {\footnotesize 0} & {\footnotesize 0} &  & {\footnotesize -1} & {\footnotesize -1} & {\footnotesize 0} & {\footnotesize$cit$}\tabularnewline
 &  & {\footnotesize 1} & {\footnotesize 0} & {\footnotesize -1} &  & {\footnotesize 0} & {\footnotesize 0} & {\footnotesize 1} & {\footnotesize 0} &  & {\footnotesize 1} & {\footnotesize 0} & {\footnotesize -1} & {\footnotesize$icit$}\tabularnewline
 &  & {\footnotesize 0} & {\footnotesize 0} & {\footnotesize 0} &  & {\footnotesize 0} & {\footnotesize 0} & {\footnotesize 0} & {\footnotesize 0} &  & {\footnotesize 0} & {\footnotesize 1} & {\footnotesize 1} & {\footnotesize$cisa$}\tabularnewline
\cline{3-3}\cline{5-5}\cline{7-7}\cline{10-10}\cline{12-12}\cline{14-14}
\end{tabular}}{\footnotesize\par}
\par\end{flushleft}

\begin{flushleft}
{\footnotesize\vspace{0.5cm}
}{\footnotesize\par}
\par\end{flushleft}

\begin{flushleft}
{\footnotesize{}%
\begin{tabular}{ll|lll|l|llll|l|lll|c}
 & \multicolumn{1}{l}{} & {\footnotesize$\mathcal{Y}_{1}$} & {\footnotesize$\mathcal{Y}_{2}$} & \multicolumn{1}{l}{{\footnotesize$YR_{3}$}} & \multicolumn{1}{l}{} & {\footnotesize$h2o$} & {\footnotesize$cit$} & {\footnotesize$\mathrm{icit}$} & \multicolumn{1}{l}{{\footnotesize$cisa$}} & \multicolumn{1}{l}{} &  &  & \multicolumn{1}{l}{} & \tabularnewline
\cline{3-3}\cline{5-5}\cline{7-7}\cline{10-10}\cline{12-12}\cline{14-14}
\multirow{4}{*}{{\footnotesize$N(2)\coloneqq\textrm{diag}(L_{2})N$}} & \multirow{4}{*}{{\footnotesize$=$}} & \foreignlanguage{british}{{\footnotesize 0}} & {\footnotesize 0} & {\footnotesize 0} & \multirow{4}{*}{{\footnotesize$=$}} & {\footnotesize 0} & {\footnotesize 0} & {\footnotesize 0} & {\footnotesize 0} & \multirow{4}{*}{{\footnotesize$\cdot$}} & \foreignlanguage{british}{{\footnotesize 0}} & {\footnotesize 1} & {\footnotesize 1} & {\footnotesize$h2o$}\tabularnewline
 &  & {\footnotesize -1} & {\footnotesize -1} & {\footnotesize 0} &  & {\footnotesize 0} & {\footnotesize 1} & {\footnotesize 0} & {\footnotesize 0} &  & {\footnotesize -1} & {\footnotesize -1} & {\footnotesize 0} & {\footnotesize$cit$}\tabularnewline
 &  & {\footnotesize 1} & {\footnotesize 0} & {\footnotesize -1} &  & {\footnotesize 0} & {\footnotesize 0} & {\footnotesize 1} & {\footnotesize 0} &  & {\footnotesize 1} & {\footnotesize 0} & {\footnotesize -1} & {\footnotesize$icit$}\tabularnewline
 &  & {\footnotesize 0} & {\footnotesize 1} & {\footnotesize 1} &  & {\footnotesize 0} & {\footnotesize 0} & {\footnotesize 0} & {\footnotesize 1} &  & {\footnotesize 0} & {\footnotesize 1} & {\footnotesize 1} & {\footnotesize$cisa$}\tabularnewline
\cline{3-3}\cline{5-5}\cline{7-7}\cline{10-10}\cline{12-12}\cline{14-14}
\end{tabular}}{\footnotesize\par}
\par\end{flushleft}

\begin{flushleft}
{\footnotesize\vspace{0.5cm}
}{\footnotesize\par}
\par\end{flushleft}

\begin{flushleft}
{\footnotesize{}%
\begin{tabular}{ll|lll|l}
 & \multicolumn{1}{l}{} & {\scriptsize$\mathcal{Y}_{1}$} & {\scriptsize$\mathcal{Y}_{2}$} & \multicolumn{1}{l}{{\scriptsize$\mathcal{Y}_{3}$}} & \tabularnewline
\cline{3-3}\cline{5-5}
\multirow{4}{*}{{\scriptsize$N=\textrm{diag}^{-1}\left(L^{T}\mathbbm{1}\right)(N(1)+N(2))$}} & \multirow{4}{*}{{\scriptsize$=$}} & \foreignlanguage{british}{{\scriptsize 0}} & {\scriptsize 1} & {\scriptsize 1} & \multirow{4}{*}{{\scriptsize$=$}}\tabularnewline
 &  & {\scriptsize -1} & {\scriptsize -1} & {\scriptsize 0} & \tabularnewline
 &  & {\scriptsize 1} & {\scriptsize 0} & {\scriptsize -1} & \tabularnewline
 &  & {\scriptsize 0} & {\scriptsize 1} & {\scriptsize 1} & \tabularnewline
\cline{3-3}\cline{5-5}
\end{tabular}}{\footnotesize\par}
\par\end{flushleft}

\begin{flushleft}
{\footnotesize\vspace{0.5cm}
}{\footnotesize\par}
\par\end{flushleft}

\begin{flushleft}
{\footnotesize{}%
\begin{tabular}{l|lll|l|llll|l|lll|l|lll|cc}
\multicolumn{1}{l}{} & {\scriptsize$\mathcal{Y}_{1}$} & {\scriptsize$\mathcal{Y}_{2}$} & \multicolumn{1}{l}{{\scriptsize$\mathcal{Y}_{3}$}} & \multicolumn{1}{l}{} & {\scriptsize$h2o$} & {\scriptsize$cit$} & {\scriptsize$\mathrm{icit}$} & \multicolumn{1}{l}{{\scriptsize$cisa$}} & \multicolumn{1}{l}{} &  &  & \multicolumn{1}{l}{} & \multicolumn{1}{l}{} &  &  & \multicolumn{1}{l}{} &  & \tabularnewline
\cline{2-2}\cline{4-4}\cline{6-6}\cline{9-9}\cline{11-11}\cline{13-13}\cline{15-15}\cline{17-17}
 & \foreignlanguage{british}{{\scriptsize 0}} & {\scriptsize 1} & {\scriptsize 1} &  & {\scriptsize 1} & {\scriptsize 0} & {\scriptsize 0} & {\scriptsize 0} &  & \foreignlanguage{british}{{\scriptsize 0}} & {\scriptsize 1} & {\scriptsize 1} & \multirow{4}{*}{{\scriptsize$+$}} & \foreignlanguage{british}{{\scriptsize 0}} & {\scriptsize 0} & {\scriptsize 0} & {\scriptsize$h2o$} & \tabularnewline
\multirow{2}{*}{{\small =}} & {\scriptsize -1} & {\scriptsize -1} & {\scriptsize 0} & \multirow{2}{*}{=} & {\scriptsize 0} & {\scriptsize$\frac{1}{2}$} & {\scriptsize 0} & {\scriptsize 0} & \multirow{2}{*}{{\footnotesize$\cdot$(}} & {\scriptsize -1} & {\scriptsize -1} & {\scriptsize 0} &  & {\scriptsize -1} & {\scriptsize -1} & {\scriptsize 0} & {\scriptsize$cit$} & \multirow{2}{*}{)}\tabularnewline
 & {\scriptsize 1} & {\scriptsize 0} & {\scriptsize -1} &  & {\scriptsize 0} & {\scriptsize 0} & {\scriptsize$\frac{1}{2}$} & {\scriptsize 0} &  & {\scriptsize 1} & {\scriptsize 0} & {\scriptsize -1} &  & {\scriptsize 1} & {\scriptsize 0} & {\scriptsize -1} & {\scriptsize$icit$} & \tabularnewline
 & {\scriptsize 0} & {\scriptsize 1} & {\scriptsize 1} &  & {\scriptsize 0} & {\scriptsize 0} & {\scriptsize 0} & {\scriptsize 1} &  & {\scriptsize 0} & {\scriptsize 0} & {\scriptsize 0} &  & {\scriptsize 0} & {\scriptsize 1} & {\scriptsize 1} & {\scriptsize$cisa$} & \tabularnewline
\cline{2-2}\cline{4-4}\cline{6-6}\cline{9-9}\cline{11-11}\cline{13-13}\cline{15-15}\cline{17-17}
\end{tabular}}{\footnotesize\par}
\par\end{flushleft}

\section{\label{sec:Notation-tables}Notation tables}

\subsubsection{Notation }

Throughout this paper, $\mathbb{R}$, $\mathbb{R}^{n}$, and $\mathbb{R}^{m\times n}$
are the field of real numbers, the vector space of $n$-tuples of
real numbers, and the space of $m\times n$ matrices with entries
in $\mathbb{R}$, respectively. Similarly, $\text{\ensuremath{\mathbb{Z}}}$,
$\mathbb{Z}^{n}$, $\mathbb{Z}^{m\times n}$ are integer numbers,
the vector space of $n$-tuples of integer number, and the space of
matrices with entries in $\mathbb{Z}$, respectively. $N^{T}$ is
the transpose of a matrix $N$ in $\mathbb{R}^{m\times n}$. $\mathbb{Z}_{+}^{n}$
and $\mathbb{Z}_{++}^{n}$ are non-negative integer $n$-tuples and
positive integer $n$-tuples in $\mathbb{Z}^{n}$, respectively. Let
$\mathbbm{1}$ be the vector of all ones. For a matrix $A\in\mathbb{R}^{m\times n}$,
$A_{i}$ and $A_{:j}$ are the $i^{th}$ row and the $j^{th}$ column
of $A$, respectively, where $i\in1,\ldots,m$ and $j\in1,\ldots,n$.
Further, $\left[\,\cdot\thinspace,\cdot\,\right]$ stands for the
horizontal concatenation operator, and $I$ denotes an identity matrix.

A calligraphic, uppercase, Roman letter, e.g., $\mathcal{A}$, denotes
a set, multiset or sequence, with $\{\cdot,\cdot\}$ denoting an unordered
pair, $(\cdot,\cdot)$ denoting an ordered pair and $(\cdot,\ldots,\cdot)$
denoting a sequence. Let $\left|\mathcal{A}\right|$ denote the cardinality
of the set $\mathcal{A}$. A multiset is a modification of the concept
of a set that, unlike a set, allows for multiple instances for each
of its elements. In a multiset $\mathcal{M}\coloneqq(\mathcal{A},f)$
, $\mathcal{A}$ is a set and $f:\mathcal{A}\rightarrow\mathbb{Z}_{+}$
is a function from $\mathcal{A}$ to the set of positive integers
giving the multiplicity of the $i^{th}$ element $\mathcal{A}_{i}$
in the multiset as the number $f(\mathcal{A}_{i})$. In multiset $\{a,a,b\}$,
the element $a$ has multiplicity 2, and $b$ has multiplicity 1.
The cardinality of a multiset is constructed by summing up the multiplicities
of all its elements. The cardinality of sets, multisets and sequences
is all assumed to be finite.

In illustrative examples, all metabolic species and reactions are
annotated with their abbreviated identifier used in the Virtual Metabolic
Human database (\href{http://vmh.life}{http://vmh.life}), e.g., the
$crn$ abbreviation for the molecular species L-carnitine (\href{https://vmh.life/\#metabolite/crn}{crn}).

\begin{table}
\begin{tabular}{|c|c|}
\hline 
Symbol & Name\tabularnewline
\hline 
\hline 
$\mathcal{H}$ & directed stoichiometric hypergraph\tabularnewline
\hline 
$\mathcal{V}$ & molecular species\tabularnewline
\hline 
$\mathcal{Y}\coloneqq\{\mathcal{S}(\mathcal{\mathcal{V}}),\mathcal{P}(\mathcal{V})\}$ & reaction hyperedge\tabularnewline
\hline 
$\mathcal{S}(\mathcal{\mathcal{V}})$ & substrate chemical complex\tabularnewline
\hline 
$\mathcal{P}(\mathcal{V})$ & product chemical complex\tabularnewline
\hline 
$\mathcal{X}$ & vertex (atom)\tabularnewline
\hline 
$\mathcal{B}$ & edge (chemical bond)\tabularnewline
\hline 
$\mathcal{G}(\mathcal{X},\mathcal{B})$ & molecular graph\tabularnewline
\hline 
$\mathcal{C}(\mathcal{V})$ & chemical complex\tabularnewline
\hline 
$\mathcal{E}:=\{\mathcal{X}_{i},\,\mathcal{X}_{j}\}$ & atom transition edge\tabularnewline
\hline 
$\mathcal{G}(\mathcal{X},\mathcal{\mathcal{E}},\mathcal{H}\{\mathcal{S}(\mathcal{\mathcal{V}}),\mathcal{P}(\mathcal{V})\})$ & atom mapping\tabularnewline
\hline 
$\mathcal{T}(\mathcal{X},\mathcal{\mathcal{E}},\mathcal{H})$ & atom transition graph\tabularnewline
\hline 
$\mathcal{L}(\mathcal{X},(\mathcal{B},\mathcal{E}),(\mathcal{V},\mathcal{H}))$ & molecular transition graph\tabularnewline
\hline 
$\bar{\mathcal{B}}$ & a set of conserved bonds\tabularnewline
\hline 
$\hat{\mathcal{B}}$ & a set of reacting bonds\tabularnewline
\hline 
$\bar{\mathcal{L}}$$(\mathcal{X},(\mathcal{\bar{\mathcal{B}}},\mathcal{E}),(\mathcal{V},\mathcal{H})))$ & conserved molecular transition graph\tabularnewline
\hline 
$\hat{\mathcal{\mathcal{L}}}(\mathcal{X},(\mathcal{\hat{\mathcal{B}}},\mathcal{E}),(\mathcal{V},\mathcal{H})))$ & reacting molecular transition graph\tabularnewline
\hline 
$\mathcal{M}(\mathcal{X},\mathcal{\mathcal{E}},\mathcal{H},\mathcal{A})$ & moiety transition graph\tabularnewline
\hline 
$\mathcal{Q}(\mathcal{X},\mathcal{B})$ & conserved moiety molecular graph\tabularnewline
\hline 
$\underline{\mathcal{X}}$ & set of condensed nodes\tabularnewline
\hline 
$\underline{\hat{\mathcal{L}}}$ & condensed reacting molecular transition graph\tabularnewline
\hline 
\end{tabular}

\caption{\label{tab:Graph-theory-notation}Graph theory notation for chemical
network modelling}
\end{table}

\begin{table}
\begin{tabular}{|c|c|c|}
\hline 
Symbol & Name & Dimension\tabularnewline
\hline 
\hline 
$N$ & stoichiometric matrix & $m\times n$\tabularnewline
\hline 
$F$ & forward stoichiometric matrix & $m\times n$\tabularnewline
\hline 
$R$ & reverse stoichiometric matrix & $m\times n$\tabularnewline
\hline 
$B$ & molecular graph incidence matrix & $p\times q$\tabularnewline
\hline 
$w$ & weight vector & $q\times1$\tabularnewline
\hline 
$S$ & substrate matrix & $p\times max(q,q'$\tabularnewline
\hline 
$P$ & product matrix & $p\times max(q,q')$\tabularnewline
\hline 
$w_{s}$ & substrate weight vector & $max(q,q')\times1$\tabularnewline
\hline 
$w_{p}$ & product weight vector & $max(q,q')\times1$\tabularnewline
\hline 
$D$ & reaction matrix & $p\times max(q,q')$\tabularnewline
\hline 
$T$ & incidence matrix of an atom transition graph & $p\times t$\tabularnewline
\hline 
$V$ & matrix that maps each molecular species to each atom & $m\times p$\tabularnewline
\hline 
$E$ & matrix that maps each directed atom transitions to each reaction & $t\times n$\tabularnewline
\hline 
$A$ & incidence matrix of a molecular transition graph & $p\times q$\tabularnewline
\hline 
$C$ & mapping between connected components & $c\times p$\tabularnewline
\hline 
$P$ & permutation matrix & \tabularnewline
\hline 
$H$ & mapping between isomorphism classes & $\left|\mathcal{I}\right|\times c$\tabularnewline
\hline 
$M$ & incidence matrix of the moiety transition graph & $u\times v$\tabularnewline
\hline 
$L$ & conserved moiety matrix & $(m-r)\times m$\tabularnewline
\hline 
$\bar{\underline{A}}$ & incidence matrix of the condensed conserved graph & \tabularnewline
\hline 
\end{tabular}

\caption{\label{tab:Matrix-notations-for}Matrix notations for chemical network
modelling}
\end{table}

\begin{center}
\begin{table}
\begin{tabular}{|c|l|}
\hline 
Symbol & Name\tabularnewline
\hline 
\hline 
$m$ & number of metabolites\tabularnewline
\hline 
$n$ & number of reactions\tabularnewline
\hline 
$p$ & number of nodes\tabularnewline
\hline 
$q$ & number of edges\tabularnewline
\hline 
$d$ & number of reacting bonds\tabularnewline
\hline 
$p$($\mathcal{V}_{k}$) & cardinality of atoms of a molecular species\tabularnewline
\hline 
$q$($\mathcal{V}_{k}$) & cardinality of bonds of a molecular species\tabularnewline
\hline 
$\alpha(\mathcal{V}_{k})$ & atomic cardinality of molecular species\tabularnewline
\hline 
$t\coloneqq\left|\mathcal{E}\right|$ & number of atom mappings\tabularnewline
\hline 
$\mathcal{I}$ & number of maximal isomorphism classes\tabularnewline
\hline 
$u$ & number of conserved moiety instances\tabularnewline
\hline 
$v$ & number of conserved moiety transitions\tabularnewline
\hline 
$s\coloneqq\mid\mathcal{\mathcal{Z}}\mid$ & number of bond mappings\tabularnewline
\hline 
\end{tabular}\caption{\label{tab:Variable-notations-for}Variable notations for chemical
network modelling}
\end{table}
\par\end{center}
\end{document}